\documentclass[journal]{IEEEtran}

\ifCLASSINFOpdf
\usepackage[pdftex]{graphicx}
\DeclareGraphicsExtensions{.pdf,.jpeg,.png}
\else
\usepackage[dvips]{graphicx}
\fi
\usepackage{epstopdf}
\usepackage{amssymb,amsmath,amsthm}
\usepackage{amssymb}
\usepackage{wasysym}
\usepackage{algorithm}
\usepackage{algorithmic}
\usepackage{array}
\usepackage{cite}
\usepackage{color}
\usepackage{url}

\usepackage{epsfig,latexsym}
\usepackage{flushend}
\usepackage{verbatim}
\usepackage{amsopn}
\usepackage{booktabs}

\usepackage{stfloats}
\usepackage{enumerate}
\usepackage{hyperref}
\usepackage{subfigure}
\usepackage{caption}
\usepackage{bm}
\newtheorem{theorem}{Theorem}

\renewcommand{\baselinestretch}{0.812}

\begin{document}
	
	\title{Movable Antenna Enhanced Integrated Sensing and Communication Via Antenna Position Optimization}
	
	\author{{Wenyan Ma, \IEEEmembership{Graduate Student Member, IEEE}, Lipeng Zhu, \IEEEmembership{Member, IEEE}, and  Rui Zhang, \IEEEmembership{Fellow, IEEE}}
		\vspace{-25pt}
		
		\thanks{	
			W. Ma and L. Zhu are with the Department of Electrical and Computer Engineering, National University of Singapore,
			Singapore 117583 (Email: {wenyan@u.nus.edu}, {zhulp@nus.edu.sg}).
			
			R. Zhang is with the School of Science and Engineering, Shenzhen Research Institute of Big Data, The Chinese University of Hong Kong, Shenzhen, Guangdong 518172, China (e-mail: rzhang@cuhk.edu.cn). He is also with the Department of Electrical and Computer Engineering, National University of Singapore, Singapore 117583 (e-mail: elezhang@nus.edu.sg).
	}}
	\maketitle
	
	\begin{abstract}
		In this paper, we propose an integrated sensing and communication (ISAC) system aided by the movable-antenna (MA) array, which can improve the communication and sensing performance via flexible antenna movement over conventional fixed-position antenna (FPA) array. First, we consider the downlink multiuser communication, where each user is randomly distributed  within a given three-dimensional zone with local movement. To reduce the overhead of frequent antenna movement, the antenna position vector (APV) is designed based on users' statistical channel state information (CSI), so that the antennas only need to be moved in a large timescale. Then, for target sensing, the Cramer-Rao bounds (CRBs) of the estimation mean square error for different spatial angles of arrival (AoAs) are derived as functions of MAs' positions. Based on the above, we formulate an optimization problem to maximize the expected minimum achievable rate among all communication users, with given constraints on the maximum acceptable CRB thresholds for target sensing. An alternating optimization algorithm is proposed to iteratively optimize one of the horizontal and vertical APVs of the MA array with the other being fixed. Numerical results demonstrate that our proposed MA arrays can significantly enlarge the trade-off region between communication and sensing performance compared to conventional FPA arrays with different inter-antenna spacing. It is also revealed that the steering vectors of the designed MA arrays exhibit low correlation in the angular domain, thus effectively reducing channel correlation among communication users to enhance their achievable rates, while alleviating ambiguity in target angle estimation to achieve improved sensing accuracy.
		
	\end{abstract}
	
	\section{Introduction}
	
	The next generation of mobile communication systems is anticipated to enable a wide range of location-aware applications, including autonomous driving, robotic navigation, and virtual reality \cite{jiang2021the,chowdhury20206g}. These applications require wireless networks to provide advanced sensing capabilities, beyond the traditional quality of service (QoS) requirement on data transmission rates and reliability. As a result, there is a growing interest in integrated sensing and communication (ISAC), which is a new paradigm that combines sensing and communication functionalities by utilizing shared hardware and radio resources. In the context of ISAC, sensing involves extracting essential information about targets and the surrounding environment. For the implementation of ISAC, multiple-input multiple-output (MIMO) technology is widely recognized as a key enabler, offering advanced precoding capabilities for spatial adaptation and waveform shaping \cite{liu2022survey,shao2024intelligent}.
	
	To achieve enhanced spatial multiplexing for communication and form sensing beams with high angular resolution, ISAC transceivers typically require large-scale antenna arrays \cite{mailloux2005phased,wirth2005radar}. However, the associated hardware cost and power consumption increase proportionally with the number of antennas, bringing a significant challenge in developing cost-effective and high-performance ISAC systems. To address this issue, sparse antenna arrays have been proposed as a cost-efficient solution, where the number of antennas is reduced by enlarging the inter-antenna spacing to approximate the angular resolution of large-scale arrays \cite{greene1978sparse,roberts2011sparse}. Despite their advantages, sparse arrays typically rely on fixed-position antennas (FPAs), which lack the flexibility to adapt to varying communication and sensing requirements in wireless networks. This limitation prevents them from dynamically switching between the optimal array geometries for communication and sensing tasks \cite{wang2023can,gazzah2009optimum}. Additionally, FPAs in both large-scale and sparse arrays are unable to fully exploit the variations in wireless channels within the spatial region where the ISAC transmitter or receiver is located.
	
	To address the limitations of conventional FPA-based ISAC systems, we investigate in this paper ISAC systems aided by movable-antenna (MA) arrays \cite{zhu2023MAMag} (also known as the fluid antenna system (FAS) in the literature \cite{zhu2024historical,new2024tutorial,wutuo2024fluid,yeyuqi2023fluid}), which enable the flexible adjustment of antennas' positions at the ISAC transmitter/receiver. By introducing the new degree of freedom (DoF) in antenna position optimization, MA-aided ISAC systems can offer significant potential to enhance communication and sensing performance while maintaining the same number of antennas as conventional FPA array, explained as follows. First, by expanding the antenna movement region, MA array effectively increases its aperture, which not only improves angular resolution for precise angle estimation \cite{ma2024MAsensing}, but also provides enhanced spatial multiplexing performance for multiuser communications \cite{zhu2023MAmultiuser}. Additionally, the geometry of an MA array can be optimized to reduce the correlation between steering vectors across different directions, thereby decreasing ambiguity in angle estimation for wireless sensing and mitigating interference for communication users in different directions. Moreover, the real-time adjustability of MAs' positions enables the ISAC system to adapt to time-varying environmental conditions and diverse communication and sensing requirements. In practice, the geometry of an MA array can either be pre-configured for specific communication/sensing applications or dynamically adjusted in real time to adapt to varying environments and ISAC requirements.
	
	There has been a resurgence of research interest in MA recently due to its new applications and benefits unveiled in wireless communication. For example, in \cite{zhu2022MAmodel,mei2024movable,zhu2024wideband,new2024tutorial}, it was shown that utilizing MAs can effectively enhance the received signal-to-noise ratio (SNR) under either deterministic or stochastic channel model. The MA-aided multiuser communication systems have been widely investigated \cite{zhu2023MAmultiuser,xiao2023multiuser,wu2023movable,qin2024antenna,cheng2023sum,yang2024flexible,wutuo2024fluid}, where the MA position optimization can help mitigate the multiuser interference. In \cite{ma2022MAmimo,chen2023joint,yeyuqi2023fluid}, the spatial
	multiplexing of MA-aided MIMO systems were characterized. The channel estimation techniques for MA systems were explored in \cite{ma2023MAestimation,xiao2023channel}, reconstructing the channel response for arbitrary transmit and receive antenna locations. Moreover, In \cite{zhu2023MAarray,ma2024multi}, the efficacy of MA arrays was demonstrated in interference nulling and multi-beamforming. The
	efficiency of MA arrays in satellite communication and secure transmission were also
	studied in \cite{ZhuLP_satellite_MA} and \cite{hu2024secure,tang2024secure}, respectively. Furthermore, the six-dimensional MA (6DMA) system was introduced in \cite{shao20246DMA,shao2024discrete,shao2024Mag6DMA,shao2024exploiting}, which further incorporates three-dimensional (3D) antenna rotation to fully exploit the spatial DoFs in antenna movement.
	
	More recently, an increasing research interest has been drawn to the field of MA-aided ISAC systems. The sensing performance metrics involve maximizing the radar beamforming gain  \cite{WuHS_MA_RIS_ISAC,hao2024fluid,kuang2024movableISAC,zhang2024efficient,mayaodong2024movable,khalili2024advanced} or signal-to-clutter-plus-noise ratio (SCNR) \cite{lyu2024flexibleISAC,peng2024jointISAC,wang2024multiuser,xiu2024movable} with given target angle of arrival (AoA), and minimizing the lower-bound on the AoA estimation mean square error (MSE), i.e., the Cramer-Rao bound (CRB), for the given AoA \cite{qin2024cramer}. Specifically, the authors in \cite{WuHS_MA_RIS_ISAC,hao2024fluid} aimed to maximize the beamforming gain towards the target AoA, subject to the given constraint on the minimum signal-to-interference-plus-noise ratio (SINR) of all communication users. In \cite{kuang2024movableISAC,zhang2024efficient,mayaodong2024movable}, the multiuser sum-rate was maximized, with given constraint on the minimum beamforming gain towards the target. In \cite{khalili2024advanced}, the authors aimed to minimize the total transmit power of the ISAC transmitter, while guaranteeing the minimum SINR of each communication user and the required MSE between	the desired and actual beamforming gains towards the target AoAs. Furthermore, in \cite{lyu2024flexibleISAC,peng2024jointISAC}, the authors aimed to maximize the weighted sum of communication sum-rate and the sensing mutual information (MI), which is a function of the radar SCNR. In \cite{wang2024multiuser,xiu2024movable}, the downlink sum-rate and SNR were maximized, respectively, under the given sensing SCNR constraint. Moreover, in \cite{qin2024cramer}, the authors minimized the sensing CRB with given AoA information by jointly optimizing transmit beamforming and the positions of MAs at both users and the base station (BS), while ensuring a minimum SINR for communication users. However, the above studies focus on the sensing performance with given target AoA, while the fundamental relationship between MAs' positions and ISAC performance in target sensing without prior AoA information has yet to be explored.
	
	In this paper, we consider the MA-aided ISAC system by leveraging the additional DoFs offered by antenna position optimization. However, different from the aforementioned prior works \cite{WuHS_MA_RIS_ISAC,hao2024fluid,kuang2024movableISAC,zhang2024efficient,mayaodong2024movable,khalili2024advanced,lyu2024flexibleISAC,peng2024jointISAC,wang2024multiuser,xiu2024movable,qin2024cramer}, we design MAs' positions to enhance communication and sensing performance based on the statistical knowledge of communication users' channels and without relying on any prior knowledge of the sensing target's AoA. The main results of this paper are summarized as follows:
	
	\begin{itemize}
		\item First, we consider the downlink multiuser communication, where each user is randomly distributed  within a given 3D zone with local movement. To ensure user fairness in the long term, we aim to maximize the expected minimum achievable rate among all users by averaging out their channel variations due to local movement. In addition, to reduce the overhead of frequent antenna movement, the antenna position vector (APV) is designed based on users' statistical channel state information  (CSI) given their location zones. On the other hand, for target sensing, we derive the CRB of AoA estimation MSE as a function of MAs' positions. Specifically, we consider the maximum CRBs of the estimation MSE for the two AoAs with respect to (w.r.t.) the horizontal and vertical axes, respectively.
		\item Next, we formulate an optimization problem to maximize the expected minimum achievable rate among all communication users, with given constraints on the maximum acceptable CRB thresholds for target sensing. To solve this problem efficiently, we consider zero-forcing (ZF) precoding for communication users to obtain high-quality suboptimal solutions with low computational complexity. Then, an alternating optimization algorithm is proposed to iteratively optimize one of the horizontal and vertical APVs of the MA array with the other being fixed.
		\item Finally, extensive numerical results are presented to compare the trade-off region  between communication and sensing performance achieved by MA arrays versus conventional FPA arrays with different inter-antenna spacing. Important insights are also provided on how the designed MAs' positions can improve both the communication and sensing performance over FPA arrays. It is also revealed that the steering vectors of the designed MA arrays exhibit low correlation in the angular domain, thus effectively reducing channel correlation among communication users to enhance their achievable rates, while alleviating ambiguity in target angle estimation to achieve improved sensing accuracy.
	\end{itemize}
	
	The rest of this paper is organized as follows. Section II introduces the system model and formulates the optimization problem to characterize the performance trade-offs between communication and sensing in MA-aided ISAC systems. Section III presents the alternating optimization algorithm to solve the formulated problem. Numerical results and relevant discussions are provided in Section IV. Finally, the conclusions are drawn in Section V.

	\textit{Notations}: Vectors and matrices are represented using boldface lowercase and uppercase symbols, respectively. The operations of conjugate, transpose, and conjugate transpose are denoted by $(\cdot)^{\mathsf *}$, $(\cdot)^{\mathsf T}$, and $(\cdot)^{\mathsf H}$, respectively. The sets of $P \times Q$ complex-valued and real-valued matrices are represented by $\mathbb{C}^{P \times Q}$ and $\mathbb{R}^{P \times Q}$, respectively. The $p$th entry of a vector $\bm{a}$ is expressed as $\bm{a}[p]$, and the entry of a matrix $\bm{A}$ in its $p$th row and $q$th column is denoted by $\bm{A}[p,q]$. The $N$-dimensional identity matrix and the column vector with all elements equal to $1$ are denoted by $\bm{I}_N$ and $\bm{1}_N$, respectively. The ceiling of a real number $a$ is written as $\lceil a \rceil$. The $2$-norm of a vector $\bm{a}$ and the Frobenius norm of a matrix $\bm{A}$ are denoted by $\|\bm{a}\|_2$ and $\|\bm{A}\|_{\rm F}$, respectively. $\mathcal{CN}(0,\bm{\Gamma})$ denotes the circularly symmetric complex Gaussian (CSCG) distribution with mean zero and covariance matrix $\bm{\Gamma}$. The trace of a matrix $\bm{A}$ is denoted by ${\rm{Tr}}(\bm{A})$. Finally, $\otimes$ represents the Kronecker product.

	\begin{figure}[!t]
		\centering
		\includegraphics[width=70mm]{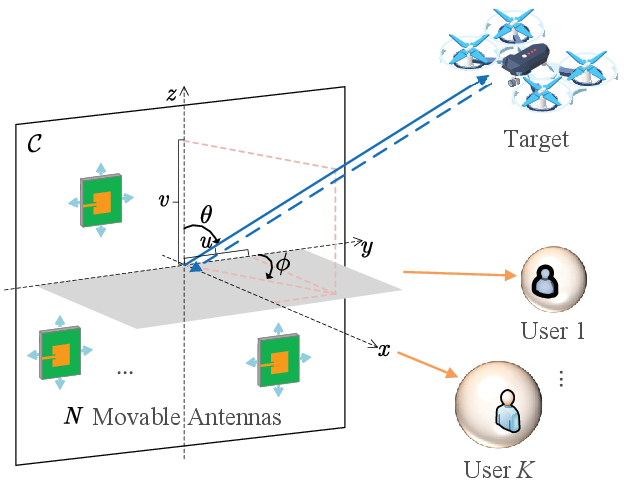}
		\caption{The considered MA-aided ISAC system.}
		\label{system}
	\end{figure}

	\section{System Model}
	As shown in Fig.~\ref{system}, we consider an ISAC system where a BS has $N$ MAs deployed on the $y\text{-}z$ plane to serve $K$ single-FPA users. To realize wireless sensing, the BS simultaneously transmits and receives the probing signals to estimate the target's AoAs w.r.t. $y$ and $z$ axes. The two-dimensional (2D) region for antenna movement is assumed to be continuous and denoted by $\mathcal{C}$, where the 2D coordinate of the $n$th ($n=1,2,\ldots,N$) MA is denoted as $\bm{r}_n=[y_n,z_n]^{\mathsf T}\in\mathcal{C}$. Denote the APV of $N$ MAs' positions by $\tilde{\bm{r}}=\left[\bm{r}_1^{\mathsf T}, \bm{r}_2^{\mathsf T}, \ldots, \bm{r}_N^{\mathsf T}\right]^{\mathsf T} \in{\mathbb{R}^{2N\times1}}$. To avoid self-interference between the communication and sensing subsystems, we assume that they share the same hardware but operate at separate time or frequency resource blocks. This allows for flexible design of the transmit signals/waveforms for each subsystem. However, the positions of the MAs must be properly designed to balance the trade-off between communication and sensing performance.

	\subsection{Communication Subsystem}
	MAs can enhance spatial multiplexing performance in multiuser communication systems via antenna position optimization. As shown in Fig.~\ref{system}, we consider the downlink multiuser communications, where each user is randomly distributed  within a given 3D zone $\mathcal{G}_k$ by accounting for its local movement, and receives signals from the BS. For ease of exposition, we assume a line-of-sight (LoS) channel between any user location and the BS \footnote{Since ISAC systems are typically operated in high frequency bands, such as millimeter-wave bands, the channel between the BS and user is generally dominated by the LoS path \cite{lirenwang2024irs}. In addition, the LoS channel model can be extended to a more general field-response channel model with multiple paths \cite{zhu2023MAMag}, and the optimization algorithms proposed in this paper can also be extended to this case.}. Denote the 3D location of the $k$th ($k=1,2,\ldots,K$) user and the $n$th ($n=1,2,\ldots,N$) MA as $\bm{u}_k=[u_{k,x},u_{k,y},u_{k,z}]^{\mathsf T}\in\mathcal{G}_k$ and $\bm{r}_n^{\rm{3D}}=[0,y_n,z_n]^{\mathsf T}$, respectively. Then, the distance and 3D wave vector from the BS to the $k$th user are given by $d(\bm{u}_k)=\|\bm{u}_k\|_2$ and $\bm{n}(\bm{u}_k)=\bm{u}_k/d(\bm{u}_k)$, respectively. Let $\bar{\bm{n}}(\bm{u}_k) = [u_{k,y}/d(\bm{u}_k), u_{k,z}/d(\bm{u}_k)]^{\mathsf{T}}$ denote the 2D wave vector to the $k$th user, which is the projection of $\bm{n}(\bm{u}_k)$ onto the $y\text{-}z$ plane. Accordingly, the propagation distance difference for the signal to the $k$th user between the MA’s position $\bm{r}_n^{\rm{3D}}$ and the origin of the $y\text{-}z$ plane can be expressed as
	\begin{equation}
		\bm{n}(\bm{u}_k)^{\mathsf T} \bm{r}_n^{\rm{3D}} = \bar{\bm{n}}(\bm{u}_k)^{\mathsf T} \bm{r}_n.
	\end{equation}
	Thus, the steering vector of the 2D MA array can be expressed as a function of the MAs' positions $\tilde{\bm{r}}$ and the 2D wave vector $\bar{\bm{n}}(\bm{u}_k)$ as follows:
	\begin{equation}\label{steering}
		\bm{\alpha}(\tilde{\bm{r}},\bar{\bm{n}}(\bm{u}_k)) \triangleq \left[ e^{j\frac{2\pi}{\lambda}\bar{\bm{n}}(\bm{u}_k)^{\mathsf T} \bm{r}_1}, \ldots, e^{j\frac{2\pi}{\lambda}\bar{\bm{n}}(\bm{u}_k)^{\mathsf T} \bm{r}_N} \right]^{\mathsf T} \in{\mathbb{C}^{N\times1}},
	\end{equation}	
	where $\lambda$ is the carrier wavelength. As a result, the LoS channel for the $k$th user at location $\bm{u}_k$ is expressed as
	\begin{align}\label{hk}
		\bm{h}_{k}(\tilde{\bm{r}},\bm{u}_k) &= \frac{\lambda }{4\pi d(\bm{u}_k)} e^{j\frac{2\pi}{\lambda}d(\bm{u}_k)} \bm{\alpha}(\tilde{\bm{r}},\bar{\bm{n}}(\bm{u}_k)).
	\end{align}
	For the considered MA-aided multiuser communication	system, multiple users are served by the BS via space-division multiple access (SDMA) in the downlink for maximizing the spatial multiplexing gain. Specifically, denote the collection of $K$ users' locations by $\tilde{\bm{u}}=\left[\bm{u}_1, \bm{u}_2, \ldots, \bm{u}_K\right] \in{\mathbb{R}^{3\times{K}}}$, and the channel	matrix between the BS and all users by $\bm{H}(\tilde{\bm{r}},\tilde{\bm{u}}) = \left[\bm{h}_{1}(\tilde{\bm{r}},\bm{u}_1), \bm{h}_{2}(\tilde{\bm{r}},\bm{u}_2), \ldots, \bm{h}_{K}(\tilde{\bm{r}},\bm{u}_K)\right] \in{\mathbb{C}^{N\times{K}}}$. Then, the received signals at all users are expressed as
	\begin{align}
		\bm{y} = \bm{H}(\tilde{\bm{r}},\tilde{\bm{u}})^{\mathsf H}\bm{W}\bm{s} + \bm{n},
	\end{align}	
	where $\bm{W}=[\bm{w}_1,\bm{w}_2,\ldots,\bm{w}_K]\in{\mathbb{C}^{N\times{K}}}$ represents the precoding matrix at the BS with a maximum power constraint $\|\bm{W}\|_{\rm F}^2\leq P$; $\bm{s}\in \mathbb{C}^{K\times1}$ is the transmit signal vector at the BS satisfying $\mathbb{E}\{\bm{s}\bm{s}^{\mathsf H}\} = \bm{I}_K$; and $\bm{n} \sim \mathcal{CN}(0,\sigma^2\bm{I}_K)$ denotes the additive white Gaussian noise (AWGN) vector at all users, with an average power of $\sigma^2$. Consequently, the receive SINR for the $k$th user can be expressed as
	\begin{align}
		\gamma_{k}(\tilde{\bm{r}},\bm{W},\bm{u}_k)= \frac{|\bm{w}_k^{\mathsf H}\bm{h}_{k}(\tilde{\bm{r}},\bm{u}_k)|^2}{\sum_{i=1,i\neq k}^{K}|\bm{w}_i^{\mathsf H}\bm{h}_{k}(\tilde{\bm{r}},\bm{u}_k)|^2 + \sigma^2}.
	\end{align}
	Accordingly, the achievable rate of the $k$th user is given by
	\begin{align}
		R_k(\tilde{\bm{r}},\bm{W},\bm{u}_k) = \log_2\left( 1+\gamma_{k}(\tilde{\bm{r}},\bm{W},\bm{u}_k) \right).
	\end{align}
	
	The wireless channels between the BS and users vary over time due to each user's local movement within each given zone. However, the limited movement speed of the MA array may not able to adapt to the instantaneous channels for mobile users. To address this practical issue, we propose a two-timescale design scheme for MA arrays: the precoding matrix $\bm{W}$ is optimized based on instantaneous channels of users, while the APV is designed based on the statistical CSI of users, which, in our context, refers to the uniformly distributed LoS channels corresponding to each user's located zone. Moreover, to ensure user fairness in the long term, we aim to maximize the expected minimum achievable rate among all users, which is expressed as
	\begin{align}\label{barR}
		\bar{R}(\tilde{\bm{r}}) = \underset{\tilde{\bm{u}}}{\mathbb{E}}\{\max_{\bm{W}} \min_{k} R_k(\tilde{\bm{r}},\bm{W},\bm{u}_k) \},
	\end{align}
	where the expectation is taken over the random communication users' locations within their respective zones, $\{\mathcal{G}_k\}_{k=1}^K$. Notably, unlike conventional multiuser communication with FPAs, the expected minimum achievable rate for MA systems in \eqref{barR} depends on the positions of MAs, $\tilde{\bm{r}}$, which affect both the channel matrix $\bm{H}(\tilde{\bm{r}},\tilde{\bm{u}})$ and the corresponding precoding matrix $\bm{W}$.

	Since deriving the expectation in \eqref{barR} analytically is challenging, the Monte Carlo method can be utilized to approximate $\bar{R}(\tilde{\bm{r}})$, which involves generating $Q$ independent realizations of users' locations, and then averaging the corresponding achievable rate over all random realizations \cite{shao20246DMA}. Thus, the expected minimum achievable rate of all users in \eqref{barR} can be approximated as 
	\begin{align}\label{tildeR}
		\bar{R}(\tilde{\bm{r}})\approx\tilde{R}(\tilde{\bm{r}}) \triangleq \frac{1}{Q} \sum_{q=1}^{Q} \max_{\bm{W}^{q}} \min_{k} R_k(\tilde{\bm{r}},\bm{W}^{q},\bm{u}_k^{q}),
	\end{align}
	where $\tilde{\bm{u}}^{q}=\left[\bm{u}_1^q, \bm{u}_2^q, \ldots, \bm{u}_K^q\right] \in{\mathbb{R}^{3\times{K}}}$ with $\bm{u}_k^q \in\mathcal{G}_k$, $k=1,2,\ldots,K$, and $\bm{W}^{q}=[\bm{w}_1^q,\bm{w}_2^q,\ldots,\bm{w}_K^q]\in{\mathbb{C}^{N\times{K}}}$ denote the users' locations and the corresponding precoding matrix at the BS, respectively, for maximizing the minimum achievable rate among users under the $q$th ($q=1,2,\ldots,Q$) realization.

	\subsection{Sensing Subsystem}
	As shown in Fig.~\ref{system}, we consider a monostatic radar system with $N$ MAs deployed on a 2D plane to estimate the target's AoAs w.r.t. the $y$ and $z$ axes. To perform AoA estimation, the BS consecutively transmits the probing signals and then receives the echoes reflected by the target over $T$ snapshots. We assume an LoS channel for the BS-target-BS link, which remains static during $T$ snapshots. Given that the antenna movement region is typically much smaller than the BS-target distance, we adopt the far-field channel model for the BS-target-BS link \cite{zhu2022MAmodel,zhu2023MAmultiuser}. As depicted in Fig.~\ref{system}, the physical elevation and azimuth AoAs of the LoS path between the BS and target are denoted by $\theta \in [0, \pi]$ and $\phi \in [0, \pi]$, respectively. For convenience, the spatial AoAs w.r.t. $x$, $y$, and $z$ axes are defined as
	\begin{equation}
		w\triangleq\sin \theta \sin \phi, u\triangleq\sin \theta \cos \phi, v\triangleq\cos \theta.
	\end{equation}
	Then, the 3D wave vector from the BS to the target is expressed as $\bm{n}^{\rm{s}} = [w, u, v]^{\mathsf T}$. Let $\bm{\chi} = [u, v]^{\mathsf{T}}$ denote the 2D wave vector to the target, which is the projection of $\bm{n}^{\rm{s}}$ onto the $y\text{-}z$ plane. Accordingly, the propagation distance difference for the signal to the target between the MA’s position $\bm{r}_n^{\rm{3D}}$ and the origin of the $y\text{-}z$ plane can be expressed as
	\begin{equation}
		\left(\bm{n}^{\rm{s}}\right)^{\mathsf T} \bm{r}_n^{\rm{3D}} = \bm{\chi}^{\mathsf T} \bm{r}_n.
	\end{equation}
	Thus, the steering vector for the signal to the target can be expressed as a function of the MAs' positions $\tilde{\bm{r}}$ and the two spatial AoAs $\bm{\chi}$, i.e., $\bm{\alpha}(\tilde{\bm{r}},\bm{\chi}) \in{\mathbb{C}^{N\times1}}$. As a result, the BS-target-BS LoS channel matrix is given by
	\begin{equation}\label{H2}
		\bm{H}^{\rm{s}}(\tilde{\bm{r}},\bm{\chi})=\beta\bm{\alpha}(\tilde{\bm{r}},\bm{\chi})\bm{\alpha}(\tilde{\bm{r}},\bm{\chi})^{\mathsf T},
	\end{equation}
	where $\beta$ is the complex channel coefficient incorporating the path loss of the BS-target-BS link as well as the target radar cross section (RCS).
	
	For any given APV $\tilde{\bm{r}}$, the maximum likelihood estimation (MLE) method is adopted for the joint estimation of the spatial AoAs $u$ and $v$ \cite{kay1993fundamentals}. Specifically, the received signal at the $t$th snapshot ($t=1, 2, \ldots, T$) is given by
	\begin{equation}
		\bm{y}^{\rm{s}}_t=\bm{H}^{\rm{s}}(\tilde{\bm{r}},\bm{\chi})\bm{s}^{\rm{s}}_t + \bm{z}^{\rm{s}}_t,
	\end{equation}
	where $\bm{s}^{\rm{s}}_t\in \mathbb{C}^{N\times1}$ represents the transmit probing signal from the BS with an average power $\mathbb{E}\{\|\bm{s}^{\rm{s}}_t\|_2^2\} = P^{\rm{s}}$. $\bm{z}^{\rm{s}}_t \sim \mathcal{CN}(0,\sigma^2\bm{I}_N)$ denotes the AWGN vector at the BS receiver, with an average power of $\sigma^2$.
	
	To estimate the spatial AoAs $u$ and $v$, the received signals over $T$ snapshots are stacked into the following matrix as
	\begin{equation}\label{Y}
		\bm{Y}^{\rm{s}}\triangleq[\bm{y}^{\rm{s}}_1,\bm{y}^{\rm{s}}_2,\ldots,\bm{y}^{\rm{s}}_T]=\bm{H}^{\rm{s}}(\tilde{\bm{r}},\bm{\chi})\bm{S}^{\rm{s}} + \bm{Z}^{\rm{s}},
	\end{equation}
	where $\bm{S}^{\rm{s}}\triangleq[\bm{s}^{\rm{s}}_1,\bm{s}^{\rm{s}}_2,\ldots,\bm{s}^{\rm{s}}_T]\in \mathbb{C}^{N\times T}$ and $\bm{Z}^{\rm{s}}\triangleq[\bm{z}^{\rm{s}}_1,\bm{z}^{\rm{s}}_2,\ldots,\bm{z}^{\rm{s}}_T]\in \mathbb{C}^{N \times T}$. To ensure uniform sensing performance for any $[u, v] \in [-1, 1] \times [-1, 1]$, it is generally required that $\bm{S}^{\rm{s}}$ is a row-orthogonal matrix, i.e., $\bm{S}^{\rm{s}}(\bm{S}^{\rm{s}})^{\mathsf{H}} = \frac{P^{\rm{s}}T}{N} \bm{I}_N$, such that $\bm{S}^{\rm{s}}$ generates an omnidirectional beampattern in the angular domain for uniformly scanning targets across all possible directions \cite{shao2022target}. Note that this requires $T \geq N$ to ensure that the $N \times T$ matrix $\bm{S}^{\rm{s}}$ can be row-orthogonal. For example, one such matrix $\bm{S}^{\rm{s}}$ can be constructed as a sub-matrix consisting of the first $N$ rows of a $T \times T$ discrete Fourier transform (DFT) matrix with $T \geq N$. In this case, each entry of $\bm{S}^{\rm{s}}$ is given by
	\begin{align}
		\bm{S}^{\rm{s}}[n,t] = \sqrt{\frac{P^{\rm{s}}}{N}}e^{j\frac{2\pi}{T}(t-1)(n-1)}.
	\end{align}
	Then, the two spatial AoAs can be estimated according to the following theorem.
	\begin{theorem}
		The MLE of the two spatial AoAs is given by
		\begin{equation}\label{MLE}
			\hat{\bm{\chi}} = \arg\max_{\bar{\bm{\chi}}} \left|\left(\bm{\alpha}(\tilde{\bm{r}},\bar{\bm{\chi}})\otimes \bm{\alpha}(\tilde{\bm{r}},\bar{\bm{\chi}})\right)^{\mathsf T} {\rm{vec}}\left(\bm{S}^{\rm{s}}(\bm{Y}^{\rm{s}})^{\mathsf H}\right)\right|^2,
		\end{equation}
	which can be solved by exhaustively searching for $\bar{\bm{\chi}}=[\bar{u},\bar{v}]^{\mathsf T}$ over the interval $[-1,1]\times[-1,1]$.
	\end{theorem}
	\begin{proof}
		See Appendix A.
	\end{proof}

	Then, the AoA estimation MSE can be expressed as
	\begin{align}
		{\rm{MSE}}(u)\triangleq\mathbb{E}\{|u-\hat{u}|^2\},~ {\rm{MSE}}(v)\triangleq\mathbb{E}\{|v-\hat{v}|^2\}.
	\end{align}
	Let $\bm{y}\triangleq[y_1,y_2,\ldots,y_N]^{\mathsf T}\in \mathbb{R}^{N\times1}$ and $\bm{z}\triangleq[z_1,z_2,\ldots,z_N]^{\mathsf T}\in \mathbb{R}^{N\times1}$ denote the horizontal and vertical APVs, respectively. Denoting the mean function as $\mu(\bm{y})=\frac{1}{N}\sum_{n=1}^{N}y_n$, the variance function and covariance function are defined as ${\rm{var}}(\bm{y})\triangleq \frac{1}{N}\sum_{n=1}^{N}(y_n - \mu(\bm{y}))^2 = \frac{1}{N}\sum_{n=1}^{N}y_n^2 - \mu(\bm{y})^2$ and ${\rm{cov}}(\bm{y},\bm{z})\triangleq \frac{1}{N}\sum_{n=1}^{N}(y_n-\mu(\bm{y}))(z_n-\mu(\bm{z})) = \frac{1}{N}\sum_{n=1}^{N}y_n z_n - \mu(\bm{y})\mu(\bm{z})$, respectively. Thus, the lower-bound of ${\rm{MSE}}(u)$ and ${\rm{MSE}}(v)$, i.e., the CRB, is given by the following theorem.
		\begin{theorem}
		The CRB of AoA estimation MSE is given by
		\begin{align}\label{CRBr}
			{\rm{MSE}}(u)\geq{\rm{CRB}}_u(\tilde{\bm{r}}) &=  \frac{\sigma^2\lambda^2}{16\pi^2P^{\rm{s}}TN|\beta|^2}\frac{1}{{\rm{var}}(\bm{y})-\frac{{\rm{cov}}(\bm{y},\bm{z})^2}{{\rm{var}}(\bm{z})}}, \notag\\
			{\rm{MSE}}(v)\geq{\rm{CRB}}_v(\tilde{\bm{r}}) &=  \frac{\sigma^2\lambda^2}{16\pi^2P^{\rm{s}}TN|\beta|^2}\frac{1}{{\rm{var}}(\bm{z})-\frac{{\rm{cov}}(\bm{y},\bm{z})^2}{{\rm{var}}(\bm{y})}}.
		\end{align}
	\end{theorem}
	\begin{proof}
		See Appendix B.
	\end{proof}
	
	The results in \eqref{CRBr} show that the CRB of AoA estimation MSE depends on the MAs' positions, which affect both ${\rm{CRB}}_u(\tilde{\bm{r}})$ and ${\rm{CRB}}_v(\tilde{\bm{r}})$. Specifically, ${\rm{CRB}}_u(\tilde{\bm{r}})$ and ${\rm{CRB}}_v(\tilde{\bm{r}})$ decrease as ${\rm{var}}(\bm{y}) - \frac{{\rm{cov}}(\bm{y}, \bm{z})^2}{{\rm{var}}(\bm{z})}$ and ${\rm{var}}(\bm{z}) - \frac{{\rm{cov}}(\bm{y}, \bm{z})^2}{{\rm{var}}(\bm{y})}$ increase, respectively. Hence, the MAs' positions $\tilde{\bm{r}}$ can be optimized to jointly minimize ${\rm{CRB}}_u(\tilde{\bm{r}})$ and ${\rm{CRB}}_v(\tilde{\bm{r}})$, which can be achieved by maximizing ${\rm{var}}(\bm{y})$ and ${\rm{var}}(\bm{z})$ while minimizing ${\rm{cov}}(\bm{y}, \bm{z})$. To this end, MAs should be positioned as separately as possible in both the $y$ and $z$ directions, thereby increasing ${\rm{var}}(\bm{y})$ and ${\rm{var}}(\bm{z})$, and positioned symmetrically w.r.t. the $y$ and $z$ axes to minimize ${\rm{cov}}(\bm{y}, \bm{z})$. However, a trade-off typically exists between minimizing ${\rm{CRB}}_u(\tilde{\bm{r}})$ and ${\rm{CRB}}_v(\tilde{\bm{r}})$ due to the coupling between ${\rm{var}}(\bm{y})$, ${\rm{var}}(\bm{z})$, and ${\rm{cov}}(\bm{y}, \bm{z})$.

	As shown in \eqref{CRBr}, ${\rm{CRB}}_u(\tilde{\bm{r}})$ and ${\rm{CRB}}_v(\tilde{\bm{r}})$ decrease with increasing sensing channel coefficient power $|\beta|^2$, which is independent of the MAs' positions $\tilde{\bm{r}}$. To evaluate the impact of the MAs' positions on sensing performance, we consider the minimum sensing channel coefficient power, $\tilde{\beta}$, required for the target to be detected by the BS, corresponding to the longest BS-target distance and the smallest target RCS. Then, we aim to minimize the maximum values of ${\rm{CRB}}_u(\tilde{\bm{r}})$ and ${\rm{CRB}}_v(\tilde{\bm{r}})$ with given $\tilde{\beta}$, without any prior information on the target AoAs. The maximum values of ${\rm{CRB}}_u(\tilde{\bm{r}})$ and ${\rm{CRB}}_v(\tilde{\bm{r}})$ with given $\tilde{\beta}$ are then given by
	\begin{align}\label{barCRBr}
		\overline{{\rm{CRB}}}_u(\tilde{\bm{r}}) &=  \frac{\zeta}{{\rm{var}}(\bm{y})-\frac{{\rm{cov}}(\bm{y},\bm{z})^2}{{\rm{var}}(\bm{z})}}, \\
		\overline{{\rm{CRB}}}_v(\tilde{\bm{r}}) &=  \frac{\zeta}{{\rm{var}}(\bm{z})-\frac{{\rm{cov}}(\bm{y},\bm{z})^2}{{\rm{var}}(\bm{y})}}, \notag
	\end{align}
	where $\zeta\triangleq\frac{\sigma^2\lambda^2}{16\pi^2P^{\rm{s}}TN\tilde{\beta}}$. To ensure adequate sensing performance, we impose a maximum acceptable CRB threshold $\eta$ on $\overline{{\rm{CRB}}}_u(\tilde{\bm{r}})$ and $\overline{{\rm{CRB}}}_v(\tilde{\bm{r}})$. Based on \eqref{barCRBr}, these constraints can be simplified as
	\begin{align}
		&\overline{{\rm{CRB}}}_u(\tilde{\bm{r}})\leq \eta ~\iff~ {\rm{var}}(\bm{y})-\frac{{\rm{cov}}(\bm{y},\bm{z})^2}{{\rm{var}}(\bm{z})} \geq \bar{\eta}, \\
		&\overline{{\rm{CRB}}}_v(\tilde{\bm{r}})\leq \eta ~\iff~ {\rm{var}}(\bm{z})-\frac{{\rm{cov}}(\bm{y},\bm{z})^2}{{\rm{var}}(\bm{y})} \geq \bar{\eta}, \notag
	\end{align}
	where $\bar{\eta}\triangleq \zeta/\eta$.

	\subsection{Problem Formulation}
	To characterize the trade-off between the communication and sensing performance, in this paper, we aim to maximize the expected minimum achievable rate $\tilde{R}(\tilde{\bm{r}})$ for all communication users by jointly optimizing the precoding matrices $\{\bm{W}^{q}\}_{q=1}^Q$ for $Q$ independent realizations of users' locations as well as the MAs' positions $\tilde{\bm{r}}$ at the BS, subject to a given constraint on the maximum acceptable CRB thresholds on $\overline{{\rm{CRB}}}_u(\tilde{\bm{r}})$ and $\overline{{\rm{CRB}}}_v(\tilde{\bm{r}})$ for ensuring sensing performance. Denoting $\widetilde{\bm{W}}=\left[\bm{W}^{1}, \bm{W}^{2}, \ldots, \bm{W}^{Q}\right] \in{\mathbb{R}^{N\times{QK}}}$ as the collection of $Q$ precoding matrices, the optimization problem can thus be formulated as
	\begin{subequations}
		\begin{align}
			\textrm {(P0)} \quad \max_{\tilde{\bm{r}}} \quad & \frac{1}{Q} \sum_{q=1}^{Q} \max_{\bm{W}^{q}} \min_{k} R_k(\tilde{\bm{r}},\bm{W}^{q},\bm{u}_k^{q}) \label{P0a} \\
			\text{s.t.} \quad & {\rm{var}}(\bm{y})-\frac{{\rm{cov}}(\bm{y},\bm{z})^2}{{\rm{var}}(\bm{z})} \geq \bar{\eta}, \label{P0b}\\
			&{\rm{var}}(\bm{z})-\frac{{\rm{cov}}(\bm{y},\bm{z})^2}{{\rm{var}}(\bm{y})} \geq \bar{\eta}, \label{P0c}\\ 
			&\bm{r}_n\in\mathcal{C},~~n=1,2,\ldots,N,\label{P0d}\\
			& \|\bm{r}_n-\bm{r}_m\|_2\geq D_0,~~n\neq m,\label{P0e}\\
			& \|\bm{W}^{q}\|_{\rm F}^2\leq P,~~q=1,2,\ldots,Q, \label{P0f}
		\end{align}
	\end{subequations}
	where constraint \eqref{P0d} ensures that the MAs are moved within the feasible region $\mathcal{C}$; constraint \eqref{P0e} enforces a minimum distance of $D_0$ between adjacent MAs to prevent antenna coupling;	and the transmit power of the communication system at the BS is no larger than $P$ in constraint \eqref{P0f}. To guarantee the sensing performance, $\bar{\eta}$ in constraints \eqref{P0b} and \eqref{P0c} is a given threshold determined by the maximum	CRB constraint. Moreover, to simplify the optimization, we assume that $\mathcal{C}$ is a convex 2D region, ensuring that constraint \eqref{P0d} is convex. For cases where $\mathcal{C}$ is a non-convex 2D region, its largest convex sub-region can be identified using the iterative regional inflation method described in \cite{deits2015computing}. The 2D positions of the MAs within this convex sub-region can then be determined using the proposed algorithm in the following section.  Note that problem (P0) is a non-convex optimization problem because the objective function \eqref{P0a} is non-concave w.r.t. both $\tilde{\bm{r}}$ and $\widetilde{\bm{W}}$, and constraints \eqref{P0b}, \eqref{P0c}, and \eqref{P0e} are non-convex w.r.t. $\tilde{\bm{r}}$. Moreover, the coupling between $\tilde{\bm{r}}$ and $\widetilde{\bm{W}}$ in the objective function \eqref{P0a} significantly increases the difficulty of solving problem (P0).
	
	\section{Proposed Algorithm}
	
	The formulated problem (P0) cannot be optimally solved efficiently in general. To obtain suboptimal solutions, alternating optimization between $\tilde{\bm{r}}$ and $\widetilde{\bm{W}}$ can be applied, but it usually results in locally optimal solutions that may perform far from the optimal solution. For instance, given the precoding matrices at the BS designed based on the channel vectors between the BS and users, the positions of MAs often exhibit minimal variation between successive iterations. This occurs because channel vectors corresponding to other MAs' positions may not align well with the precoding matrices, leading to reduced effective channel gains for users and increased interference among them. To address this issue, we propose utilizing the ZF precoding design for communication users, which expresses the precoding matrices $\widetilde{\bm{W}}$ as functions of the MAs' positions $\tilde{\bm{r}}$. Specifically, in the $q$th realization of users' channels/locations, the ZF precoding matrix for maximizing the minimum achievable rate among all users is given by
	\begin{align}
		\bm{W}^{q}_{\rm ZF} = \sqrt{P}\frac{\bm{H}(\tilde{\bm{r}},\tilde{\bm{u}}^{q})(\bm{H}(\tilde{\bm{r}},\tilde{\bm{u}}^{q})^{\mathsf H}\bm{H}(\tilde{\bm{r}},\tilde{\bm{u}}^{q}))^{-1}}{\|\bm{H}(\tilde{\bm{r}},\tilde{\bm{u}}^{q})(\bm{H}(\tilde{\bm{r}},\tilde{\bm{u}}^{q})^{\mathsf H}\bm{H}(\tilde{\bm{r}},\tilde{\bm{u}}^{q}))^{-1}\|_{\rm F}},
	\end{align}
	which yields the same SINR (or SNR) for all users as
	\begin{align}
		\gamma^q(\tilde{\bm{r}},\tilde{\bm{u}}^{q}) &= \frac{P}{\|\bm{H}(\tilde{\bm{r}},\tilde{\bm{u}}^{q})(\bm{H}(\tilde{\bm{r}},\tilde{\bm{u}}^{q})^{\mathsf H}\bm{H}(\tilde{\bm{r}},\tilde{\bm{u}}^{q}))^{-1}\|_{\rm F}^2\sigma^2} \notag\\ 
		&= \frac{P}{{\rm{Tr}}((\bm{H}(\tilde{\bm{r}},\tilde{\bm{u}}^{q})^{\mathsf H}\bm{H}(\tilde{\bm{r}},\tilde{\bm{u}}^{q}))^{-1})\sigma^2}.
	\end{align}
	Thus, the minimum achievable rate of $K$ users is given by
	\begin{align}
		\min_{k} R_k(\tilde{\bm{r}},\bm{W}^{q}_{\rm ZF},\bm{u}_k^{q}) = \log_2\left( 1+\gamma^q(\tilde{\bm{r}},\tilde{\bm{u}}^{q}) \right).
	\end{align}	
	Then, problem (P0) is simplified as the following problem:
	\begin{subequations}
		\begin{align}
			\textrm {(P1)} \quad \max_{\tilde{\bm{r}}} \quad & \frac{1}{Q} \sum_{q=1}^{Q} \log_2\left( 1+\gamma^q(\tilde{\bm{r}},\tilde{\bm{u}}^{q}) \right) \label{P1a} \\
			\text{s.t.} \quad & \eqref{P0b},\eqref{P0c},\eqref{P0d},\eqref{P0e}.\notag
		\end{align}
	\end{subequations}
	To address the coupling between the two components of $\tilde{\bm{r}}$, i.e., the horizontal APV $\bm{y}$ and the vertical APV $\bm{z}$, in constraints \eqref{P0b} and \eqref{P0c}, an alternating optimization algorithm is introduced to obtain locally optimal solutions for problem (P1). Specifically, the proposed algorithm alternates between two subproblems of (P1), where one of the horizontal APV $\bm{y}$ and the vertical APV $\bm{z}$ is the optimization variable in each subproblem with the other being fixed. Next, we present the detailed algorithm for solving problem (P1).
	
	\subsection{Optimization of $\bm{y}$ with Given $\bm{z}$}
	In this subsection, our objective is optimizing the horizontal APV $\bm{y}$ with the vertical APV $\bm{z}$ being fixed. Let $\tilde{R}(\bm{y},\bm{z}) = \frac{1}{Q} \sum_{q=1}^{Q} \log_2\left( 1+\gamma^q(\tilde{\bm{r}},\tilde{\bm{u}}^{q}) \right)$ denote the objective function of problem (P1). Accordingly, the optimization problem w.r.t. $\bm{y}$ can be written as
	\begin{subequations}
		\begin{align}
			\textrm {(P2)} \quad \max_{\bm{y}} \quad & \tilde{R}(\bm{y},\bm{z}) \label{P2a} \\
			\text{s.t.} \quad & \eqref{P0b},\eqref{P0c},\eqref{P0d},\eqref{P0e}.\notag
		\end{align}
	\end{subequations}
	Since the constraints \eqref{P0b}, \eqref{P0c}, and \eqref{P0e} are non-convex w.r.t. $\bm{y}$, we relax them by leveraging the successive optimization technique.
	First, to relax the non-convex constraints \eqref{P0b} and \eqref{P0c}, we rewrite them into the standard quadratic form. Specifically, ${\rm{var}}(\bm{y})$, ${\rm{var}}(\bm{z})$, and ${\rm{cov}}(\bm{y},\bm{z})$ can be expressed as
	\begin{align}
		{\rm{var}}(\bm{y}) &\triangleq \bm{y}^{\mathsf T} \bm{B} \bm{y}, \notag\\
		{\rm{var}}(\bm{z}) &\triangleq \bm{z}^{\mathsf T} \bm{B} \bm{z}, \notag\\
		{\rm{cov}}(\bm{y},\bm{z}) &\triangleq \bm{y}^{\mathsf T} \bm{B} \bm{z},
	\end{align}
	where $\bm{B}\triangleq \frac{1}{N}\bm{I}_N-\frac{1}{N^2}\bm{1}_N\bm{1}_N^{\mathsf T}$ is a positive semi-definite (PSD) matrix. Consequently, \eqref{P0b} and \eqref{P0c} are equivalently transformed to
	\begin{subequations}
		\begin{align}
			& \bm{y}^{\mathsf T}\bm{B}\bm{y}\bm{z}^{\mathsf T}\bm{B}\bm{z} - \left(\bm{y}^{\mathsf T}\bm{B}\bm{z}\right)^2 \geq \bar{\eta}\bm{z}^{\mathsf T}\bm{B}\bm{z}, \label{xBx}\\
			& \bm{z}^{\mathsf T}\bm{B}\bm{z}\bm{y}^{\mathsf T}\bm{B}\bm{y} - \left(\bm{y}^{\mathsf T}\bm{B}\bm{z}\right)^2 \geq \bar{\eta}\bm{y}^{\mathsf T}\bm{B}\bm{y}. \label{yBy}
		\end{align}
	\end{subequations}
	Given $\bm{y}^p=[y_1^p,y_2^p,\ldots,y_N^p]^{\mathsf T}\in{\mathbb{R}^{N\times1}}$ obtained during the $p$th iteration of successive optimization, since $\bm{y}^{\mathsf T}\bm{B}\bm{y}$ is convex w.r.t. $\bm{y}$, it can be globally lower-bounded using its first-order Taylor expansion at $\bm{y}^p$ as
	\begin{align}
		\bm{y}^{\mathsf T}\bm{B}\bm{y} &\geq (\bm{y}^p)^{\mathsf T}\bm{B}\bm{y}^p + 2(\bm{y}^p)^{\mathsf T} \bm{B} (\bm{y}-\bm{y}^p) \notag\\
		&= 2(\bm{y}^p)^{\mathsf T} \bm{B} \bm{y} - (\bm{y}^p)^{\mathsf T}\bm{B}\bm{y}^p. \notag
	\end{align}
	Moreover, for a fixed $\bm{z}$, $\left(\bm{y}^{\mathsf T}\bm{B}\bm{z}\right)^2 = \bm{y}^{\mathsf T}(\bm{B}\bm{z}\bm{z}^{\mathsf T}\bm{B})\bm{y}$ is a convex quadratic function w.r.t. $\bm{y}$. Accordingly, in the $p$th iteration of successive optimization, constraints \eqref{xBx} and \eqref{yBy} are relaxed as convex quadratic constraints w.r.t. $\bm{y}$ as
	\begin{subequations}
		\begin{align}
			& \left(2(\bm{y}^p)^{\mathsf T} \bm{B} \bm{y} - (\bm{y}^p)^{\mathsf T}\bm{B}\bm{y}^p\right)\bm{z}^{\mathsf T}\bm{B}\bm{z} - \left(\bm{y}^{\mathsf T}\bm{B}\bm{z}\right)^2 \geq \bar{\eta}\bm{z}^{\mathsf T}\bm{B}\bm{z}, \label{xBx2}\\
			& \left(2(\bm{y}^p)^{\mathsf T} \bm{B} \bm{y} - (\bm{y}^p)^{\mathsf T}\bm{B}\bm{y}^p\right)\bm{z}^{\mathsf T}\bm{B}\bm{z} - \left(\bm{y}^{\mathsf T}\bm{B}\bm{z}\right)^2 \geq \bar{\eta}\bm{y}^{\mathsf T}\bm{B}\bm{y}. \label{yBy2}
		\end{align}
	\end{subequations}
	
	Furthermore, according to the Cauchy-Schwartz inequality, i.e., $\bm{u}^{\mathsf T}\bm{v}\leq\|\bm{u}\|_2\|\bm{v}\|_2$ for any two vectors $\bm{u}$ and $\bm{v}$ of equal size, a linear surrogate function that globally minorizes $\|\bm{r}_n-\bm{r}_m\|_2$ at $(\bm{r}_n^p-\bm{r}_m^p)$ is constructed by setting $\bm{u}\leftarrow\bm{r}_n-\bm{r}_m$ and $\bm{v}\leftarrow\bm{r}_n^p-\bm{r}_m^p$, i.e., 
	\begin{align}
		&\|\bm{r}_n-\bm{r}_m\|_2\geq \frac{\left(\bm{r}_n^p-\bm{r}_m^p\right)^{\mathsf T}\left(\bm{r}_n-\bm{r}_m\right)}{\|\bm{r}_n^p-\bm{r}_m^p\|_2},
	\end{align}
	where $\bm{r}_n^p=[y_n^p,z_n]^{\mathsf T}$. Thus, constraint \eqref{P0e} can be relaxed into a linear form expressed as
	\begin{align}\label{38}
		\frac{\left(\bm{r}_n^p-\bm{r}_m^p\right)^{\mathsf T}\left(\bm{r}_n-\bm{r}_m\right)}{\|\bm{r}_n^p-\bm{r}_m^p\|_2}\geq D_0,~~ 1\leq n<m\leq N,
	\end{align}
	which can further be reformulated into the standard linear constraint w.r.t. $\bm{y}$ as
	\begin{align}\label{39}
		\bm{D}\bm{y}\geq \bm{g},
	\end{align}
	where $\bm{g}\in\mathbb{R}^{N(N-1)/2\times1}$ and the non-zero elements of the sparse matrix $\bm{D}\in\mathbb{R}^{N(N-1)/2\times N}$ are defined as
	\begin{align}
		&\bm{g}[\rho(m,n)] = D_0\|\bm{r}_n^p-\bm{r}_m^p\|_2 - (z_n-z_m)^2, \notag\\
		&\bm{D}[\rho(m,n),n] = y_n^p - y_m^p, \\
		&\bm{D}[\rho(m,n),m] = y_m^p - y_n^p,~~ 1\leq n<m\leq N, \notag
	\end{align}	
	where $\rho(m,n)\triangleq (2N-n)(n-1)/2+m-n$. Hereto, in the $p$th iteration, the optimization of $\bm{y}$ is relaxed as
	\begin{subequations}
		\begin{align}
			\textrm {(P3)}~~\max_{\bm{y}} \quad & \tilde{R}(\bm{y},\bm{z}) \label{P3a}\\
			\text{s.t.} \quad & \eqref{xBx2}, \eqref{yBy2}, \eqref{P0d}, \eqref{39}. \notag
		\end{align}
	\end{subequations}
	Since the constraints of problem (P3) are convex w.r.t. $\bm{y}$, problem (P3) can be efficiently solved by using feasible direction methods \cite{clarkson2010coresets}. Specifically, in the $p$th iteration, we first solve the following ascent direction finding subproblem:
	\begin{subequations}
		\begin{align}
			\textrm {(P3-1)}~~\max_{\bm{d}} \quad & \bm{d}^{\mathsf T}\nabla_{\bm{y}}\tilde{R}(\bm{y}^p,\bm{z}) \label{P31a}\\
			\text{s.t.} \quad & \left(2(\bm{y}^p)^{\mathsf T} \bm{B} \bm{d} - (\bm{y}^p)^{\mathsf T}\bm{B}\bm{y}^p\right)\bm{z}^{\mathsf T}\bm{B}\bm{z} - \left(\bm{d}^{\mathsf T}\bm{B}\bm{z}\right)^2 \notag\\
			&~~\geq \bar{\eta}\bm{z}^{\mathsf T}\bm{B}\bm{z}, \label{P31b}\\
			& \left(2(\bm{y}^p)^{\mathsf T} \bm{B} \bm{d} - (\bm{y}^p)^{\mathsf T}\bm{B}\bm{y}^p\right)\bm{z}^{\mathsf T}\bm{B}\bm{z} - \left(\bm{d}^{\mathsf T}\bm{B}\bm{z}\right)^2  \notag\\
			&~~\geq \bar{\eta}\bm{d}^{\mathsf T}\bm{B}\bm{d}, \label{P31c}\\
			& \bm{r}_n(\bm{d})\in\mathcal{C},~~n=1,2,\ldots,N,\label{P31d}\\
			&\bm{D}\bm{d}\geq \bm{g}, \label{P31e}
		\end{align}
	\end{subequations}
	where $\bm{r}_n(\bm{d})\triangleq[\bm{d}[n],z_n]^{\mathsf T}$ and $\nabla_{\bm{y}}\tilde{R}(\bm{y}^p,\bm{z}) \in\mathbb{R}^{N\times1}$ denotes the gradient of the function $\tilde{R}(\bm{y},\bm{z})$ at the point $\bm{y}^p$. The $n$th element of $\nabla_{\bm{y}}\tilde{R}(\bm{y}^p,\bm{z})$ can be calculated as
	\begin{align}\label{gradient}
		\nabla_{\bm{y}}\tilde{R}(\bm{y}^p,\bm{z})[n] = \lim_{\xi\rightarrow0} \frac{\tilde{R}(\bm{y}^p+\xi\bm{e}_n,\bm{z}) - \tilde{R}(\bm{y}^p,\bm{z})}{\xi},
	\end{align}
	where $\bm{e}_n \in \mathbb{R}^{N\times1}$ is a vector with the $n$th entry equal to one and all other entries equal to zero. Problem (P3-1) is a convex optimization problem because constraints \eqref{P31b} and \eqref{P31c} are convex quadratic w.r.t. $\bm{d}$, while the objective function \eqref{P31a} and constraint \eqref{P31e} are linear w.r.t. $\bm{d}$. Furthermore, for typical circular or rectangular regions $\mathcal{C}$, where constraint \eqref{P31d} is convex quadratic or linear w.r.t. $\bm{d}$, problem (P4) is a convex quadratically-constrained quadratic programming (QCQP) problem. Such problems can be efficiently solved using existing optimization toolboxes, such as MATLAB's built-in fmincon function.
	
	Next, we determine the step size $\tau \in[0,1]$ by solving the following one-dimensional search problem:
	\begin{align}\label{tau}
		\tau = \arg\max_{\bar{\tau}\in[0,1]} \tilde{R}(\bm{y}^p + \bar{\tau}(\bm{d}-\bm{y}^p),\bm{z}),
	\end{align}
	which can be addressed through exhaustive search for $\bar{\tau}$ within the interval $[0,1]$. Finally, $\bm{y}^{p+1}$ is updated as
	\begin{align}\label{yp1}
		\bm{y}^{p+1} = \bm{y}^p + \tau(\bm{d}-\bm{y}^p).
	\end{align}
	
	The details of the proposed algorithm for solving problem (P3) are summarized in Algorithm~\ref{alg1}. Specifically, in step 4, the gradient $\nabla_{\bm{y}}\tilde{R}(\bm{y}^p,\bm{z})$ is computed via \eqref{gradient}, and then we solve the convex optimization problem (P3-1) to obtain $\bm{d}$ in step 5. Subsequently, we obtain the step size $\tau$ via \eqref{tau}, and then update $\bm{y}^{p+1}$ via \eqref{yp1}. The algorithm terminates when the increase in the objective value $\tilde{R}(\bm{y}^p,\bm{z})$ falls below the predefined convergence threshold $\epsilon_2$. Finally, the horizontal APV $\bm{y}$ is provided as the output in step 10.
	
	\begin{algorithm}[!t]
		\caption{Successive Optimization for Solving Problem (P3)}
		\label{alg1}
		\begin{algorithmic}[1]
			\STATE \emph{Input:} $N$, $\mathcal{C}$, $D_0$, $\epsilon_2$, $\bar{\eta}$, $\bm{z}$, $\bm{y}^0$.
			\STATE Initialization:  $p \leftarrow 0$.
			\WHILE{Increase of $\tilde{R}(\bm{y}^p,\bm{z})$ is above $\epsilon_2$}
			\STATE Compute $\nabla_{\bm{y}}\tilde{R}(\bm{y}^p,\bm{z})$ via \eqref{gradient}.
			\STATE Obtain $\bm{d}$ by solving problem (P3-1).
			\STATE Obtain $\tau$ via \eqref{tau}.
			\STATE Update $\bm{y}^{p+1}$ via \eqref{yp1}.
			\STATE $p \leftarrow p+1$.
			\ENDWHILE
			\STATE \emph{Output:} $\bm{y} \leftarrow \bm{y}^{p+1}$.
		\end{algorithmic}
	\end{algorithm}
	
	Next, we analyze the convergence of the proposed Algorithm~\ref{alg1}. In the $p$th iteration of solving problem (P3), we have $\tilde{R}(\bm{y}^p,\bm{z})\leq \tilde{R}(\bm{y}^{p+1},\bm{z})$, where the inequality is valid according to \eqref{tau} and \eqref{yp1}, and the equality is achieved by letting $\bm{y}^{p+1}=\bm{y}^p$. Consequently, the sequence $\{\tilde{R}(\bm{y}^p,\bm{z})\}_{p=0}^{\infty}$ is non-decreasing and converges to a maximum value. Therefore, the convergence of Algorithm~\ref{alg1} is guaranteed.
	
	The computational complexity of Algorithm~\ref{alg1} is analyzed as follows. Specifically, in step 4, the complexity for calculating $\nabla_{\bm{y}}\tilde{R}(\bm{y}^p,\bm{z})$ is $\mathcal{O}(QNK^2(N+K))$. In step 5, solving the convex optimization problem (P3-1) requires a complexity of $\mathcal{O}(N^4\ln(1/\varepsilon))$, where $\varepsilon$ denotes the accuracy for the interior-point method. In step 6, the complexity for one-dimensional linear searching $\tau$ is $\mathcal{O}(M_{\tau}QK^2(N+K))$, where $M_{\tau}$ is the number of discretizations of interval $[0,1]$. Let $I_{\bm{y}}$ represent the maximum number of iterations for executing steps 4-8. Accordingly, the total computational complexity of Algorithm~\ref{alg1} is $\mathcal{O}(((M_{\tau}+N)QK^2(N+K) + N^4\ln(1/\varepsilon))I_{\bm{y}})$, which is polynomial w.r.t. $N$, $K$, and $Q$.
	
	\subsection{Optimization of $\bm{z}$ with Given $\bm{y}$}
	In this subsection, we aim to optimize $\bm{z}$	in problem (P1) with given $\bm{y}$. Accordingly, the optimization problem w.r.t. $\bm{z}$ can be written as
	\begin{subequations}
		\begin{align}
			\textrm {(P4)} \quad \max_{\bm{z}} \quad & \tilde{R}(\bm{y},\bm{z}) \label{P4a} \\
			\text{s.t.} \quad & \eqref{P0b},\eqref{P0c},\eqref{P0d},\eqref{P0e}.\notag
		\end{align}
	\end{subequations}
	Since problem (P4) has a similar structure as (P2), Algorithm~\ref{alg1} can be adapted to obtain the vertical APV $\bm{z}$ by substituting $\left\{ \bm{y}, \bm{z}  \right\}$
	with $\left\{ \bm{z}, \bm{y}  \right\}$.
	
	Similarly, the monotonic convergence of Algorithm~\ref{alg1} is guaranteed when applied to solve problem (P4). Accordingly, the computational complexity is $\mathcal{O}(((M_{\tau}+N)QK^2(N+K) + N^4\ln(1/\varepsilon))I_{\bm{z}})$, where $I_{\bm{z}}$ denotes the maximum number of iterations for executing steps 4-8.
	
	\subsection{Overall Algorithm}
	With the solutions to problems (P2) and (P4) obtained above, we have the complete alternating optimization algorithm to solve (P1). The overall algorithm is summarized in Algorithm~\ref{alg2}. Specifically, in step 4, the horizontal APV $\bm{y}$ is optimized for a given vertical APV $\bm{z}$ by solving problem (P2) through successive optimization. Similarly, in step 5, the vertical APV $\bm{z}$ is optimized by solving problem (P4). These two subproblems are iteratively solved until the improvement in the objective value of \eqref{P1a} falls below a predefined convergence threshold, $\epsilon_1$.
	
	\begin{algorithm}[!t]
		\caption{Alternating Optimization for Solving Problem (P1)}
		\label{alg2}
		\begin{algorithmic}[1]
			\STATE \emph{Input:} $N$, $P$, $\sigma^2$, $\mathcal{C}$, $D_0$, $\epsilon_1$, $\epsilon_2$, $\bar{\eta}$.
			\STATE Initialize $\bm{y}$ and $\bm{z}$.
			\WHILE{Increase of the objective value in \eqref{P1a} is above $\epsilon_1$}
			\STATE Given $\bm{z}$, solve problem (P2) to update $\bm{y}$.
			\STATE Given $\bm{y}$, solve problem (P4) to update $\bm{z}$.
			\ENDWHILE
			
			\STATE \emph{Output:} $\bm{y}$, $\bm{z}$.
		\end{algorithmic}
	\end{algorithm}
	
	Next, we analyze the convergence of Algorithm~\ref{alg2}. The alternating optimization of the variables $\bm{y}$ and $\bm{z}$ ensures that the algorithm generates a non-decreasing sequence of objective values for (P1) during the iterations, which will not diverge to infinity due to the inherent limitations on the achievable rate. Since the convergence criterion of Algorithm~\ref{alg2} is defined as the inability to further increase the objective value of (P1) by optimizing $\bm{y}$ or $\bm{z}$, Algorithm~\ref{alg2} is guaranteed to converge to at least a locally optimal solution of (P1). The total computational complexity is given by $\mathcal{O}(((M_{\tau}+N)QK^2(N+K) + N^4\ln(1/\varepsilon))(I_{\bm{y}} + I_{\bm{z}})I)$, where $I$ represents the maximum number of outer iterations for steps 3-6 in Algorithm~\ref{alg2}.

	\subsection{Initialization}
	In this subsection, we propose an initialization scheme for the design of MAs' positions in Algorithm~\ref{alg2}. To satisfy the given constraints \eqref{P0b} and \eqref{P0c} on the maximum acceptable CRB thresholds for $\overline{{\rm{CRB}}}_u(\tilde{\bm{r}})$ and $\overline{{\rm{CRB}}}_v(\tilde{\bm{r}})$, we solve the following CRB minimization problem:
	\begin{subequations}
		\begin{align}
			\textrm {(P5)} \quad \max_{\tilde{\bm{r}},\tilde{\eta}} \quad & \tilde{\eta} \label{P5a} \\
			\text{s.t.} \quad & {\rm{var}}(\bm{y})-\frac{{\rm{cov}}(\bm{y},\bm{z})^2}{{\rm{var}}(\bm{z})} \geq \tilde{\eta}, \label{P5b}\\
			&{\rm{var}}(\bm{z})-\frac{{\rm{cov}}(\bm{y},\bm{z})^2}{{\rm{var}}(\bm{y})} \geq \tilde{\eta}, \label{P5c}\\ 
			&\eqref{P0d},\eqref{P0e}, \notag
		\end{align}
	\end{subequations}
	which can be efficiently solved via alternatively optimizing one of $\bm{y}$ and $\bm{z}$ with the other being fixed \cite{ma2024MAsensing}. The resulting solution is then set as the initial MAs' positions $\{\bm{y}^0,\bm{z}^0\}$.
	
	\subsection{Upper-bound of $\tilde{R}(\tilde{\bm{r}})$ and Lower-bound of $\eta$}
	Finally, we analyze the upper-bound of $\tilde{R}(\tilde{\bm{r}})$ and the lower-bound of the CRB threshold $\eta$. Specifically, the received SINR for the $k$th user achieves its upper-bound when there is no interference from other users. In this scenario, the maximal ratio transmission (MRT), defined as $\bm{w}_k = \sqrt{p_k}\bm{h}_k(\tilde{\bm{r}}, \bm{u}_k)/\|\bm{h}_k(\tilde{\bm{r}}, \bm{u}_k)\|_2$, can maximize the minimum achievable rate among all users, where $p_k$ is the transmit power allocated to the signal for user $k$. Consequently, the upper-bound of the SINR (or SNR) for the $k$th user is expressed as
	\begin{align}
		\gamma_{k}(\tilde{\bm{r}},\bm{W},\bm{u}_k)&\leq \frac{|\bm{w}_k^{\mathsf H}\bm{h}_{k}(\tilde{\bm{r}},\bm{u}_k)|^2}{\sigma^2} =\frac{p_k \|\bm{h}_{k}(\tilde{\bm{r}},\bm{u}_k)\|_2^2}{\sigma^2}\notag\\
		&=\frac{p_k \lambda^2N}{16\pi^2d(\bm{u}_k)^2\sigma^2}. 
	\end{align}
	To maximize the minimum SINR of all users, the optimal power allocation should ensure equal SINR for all users. This condition yields the following equations for determining $\{p_k\}_{k=1}^K$:
	\begin{align}
		\left\{
		\begin{aligned}
			&\frac{p_k \lambda^2N}{16\pi^2d(\bm{u}_k)^2\sigma^2} =  \bar{\gamma}(\tilde{\bm{u}}), \\
			&\sum_{k=1}^{K} p_k =  P, \\
		\end{aligned}
		\right.
	\end{align}
	where $\bar{\gamma}(\tilde{\bm{u}})$ denotes the upper bound of the minimum SINR among all users. It can be derived as
	\begin{align}
		\bar{\gamma}(\tilde{\bm{u}}) = \frac{P\lambda^2N}{16\pi^2\sigma^2}\left(\sum_{k=1}^{K} d(\bm{u}_k)^2 \right)^{-1}.
	\end{align}
	Therefore, the upper-bound of $\tilde{R}(\tilde{\bm{r}})$ is expressed as
	\begin{align}\label{upper}
		\tilde{R}(\tilde{\bm{r}})\leq \frac{1}{Q}\sum_{q=1}^{Q} \log_2(1+\bar{\gamma}(\tilde{\bm{u}}^{q})).
	\end{align}

	Furthermore, let $A^\textrm{cir}$ denote the radius of the minimum circumscribed circle of $\mathcal{C}$. The lower-bound of $\eta$ is given by \cite{ma2024MAsensing}
	\begin{align}\label{lower}
		\eta \geq \frac{\sigma^2\lambda^2}{8\pi^2P^{\rm{s}}TN\tilde{\beta}(A^\textrm{cir})^2}.
	\end{align}
	For the typical square region $\mathcal{C}^\textrm{squ}$ with size $A \times A$, we have $A^\textrm{cir} = A/\sqrt{2}$. Accordingly, the lower-bound of $\eta$ for the square region $\mathcal{C}^\textrm{squ}$ is $\frac{\sigma^2\lambda^2}{4\pi^2P^{\rm{s}}TN\tilde{\beta}A^2}$.
	
	\begin{figure}[!t]
		\centering
		\subfigure[Dense communication user zone]{
			\begin{minipage}{.47\textwidth}
				\centering
				\includegraphics[scale=.45]{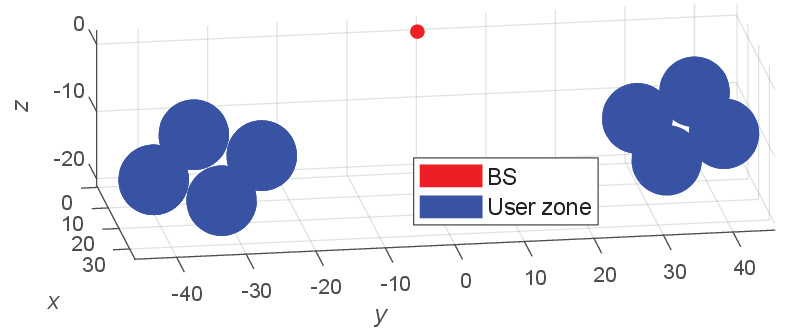}
			\end{minipage}
			\label{user_zone_dense}
		}
		\subfigure[Sparse communication user zone]{
			\begin{minipage}{.47\textwidth}
				\centering
				\includegraphics[scale=.45]{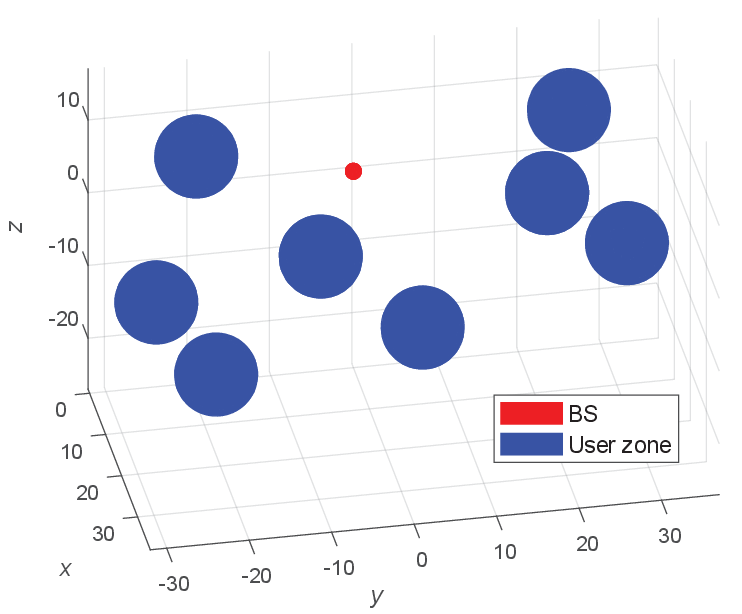}
			\end{minipage}
			\label{user_zone_sparse}
		}
		\caption{Illustration of the considered communication user zones.}
		\label{FIG2}
	\end{figure}

	\begin{figure}[!t]
		\centering
		\includegraphics[width=68mm]{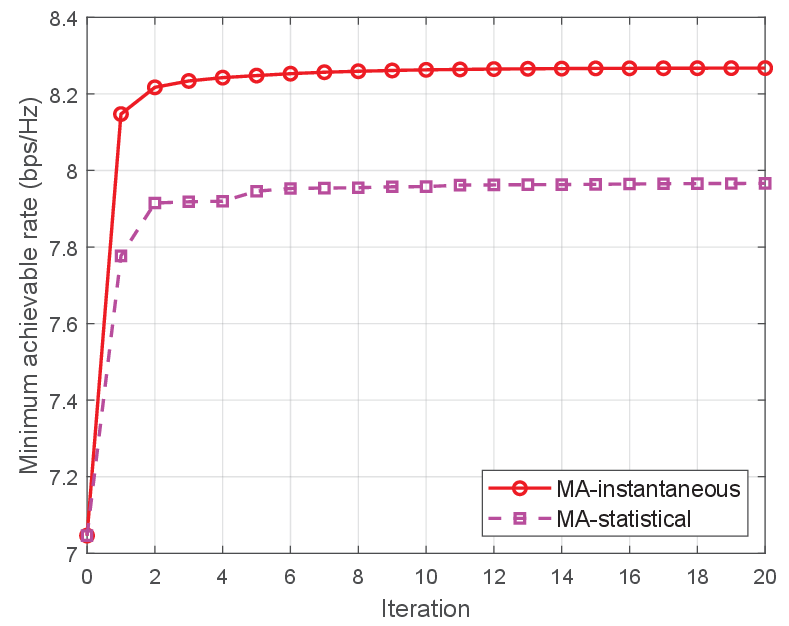}
		\caption{Convergence behavior of Algorithm~\ref{alg2}.}
		\label{convergence}
	\end{figure}
	
	\section{Numerical Results}

	\begin{figure*}[!t]
		\centering
		\subfigure[$\eta=0.02$]{
			\begin{minipage}{.31\textwidth}
				\centering
				\includegraphics[scale=.46]{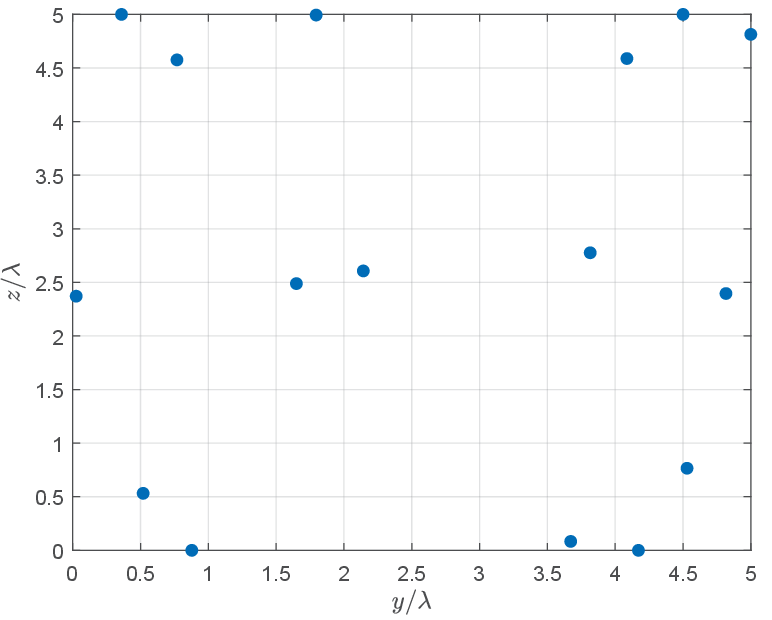}
			\end{minipage}
			\label{position1}
		}
		\subfigure[$\eta=0.002$]{
			\begin{minipage}{.31\textwidth}
				\centering
				\includegraphics[scale=.46]{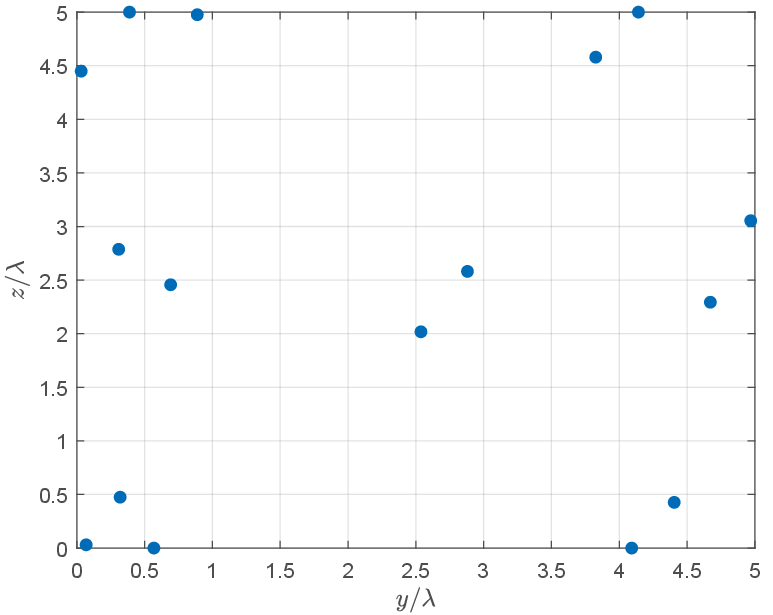}
			\end{minipage}
			\label{position2}
		}
		\subfigure[$\eta=0.001$]{
			\begin{minipage}{.31\textwidth}
				\centering
				\includegraphics[scale=.46]{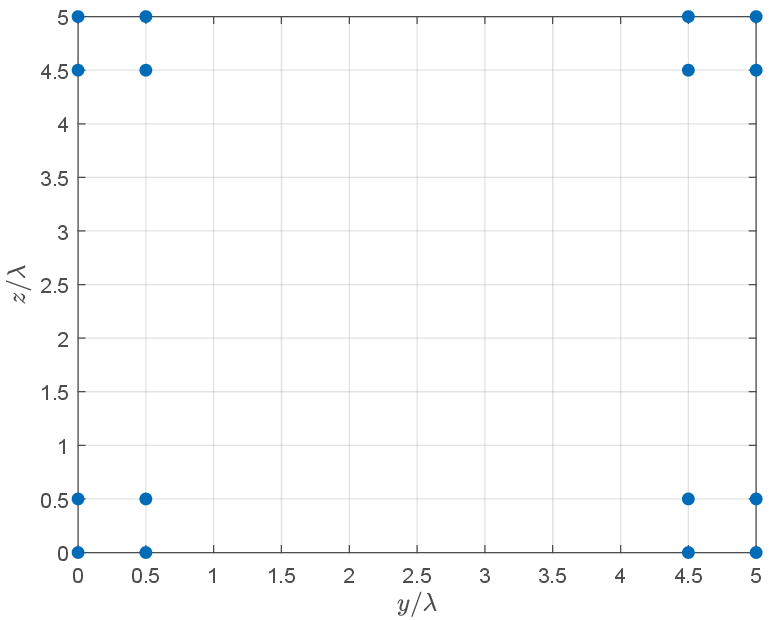}
			\end{minipage}
			\label{position3}
		}
		\caption{Illustration of the MAs’ positions for different CRB thresholds.}
		\label{FIG3}
	\end{figure*}

	\begin{figure}[!t]
		\centering
		\subfigure[Dense communication user zone]{
			\begin{minipage}{.47\textwidth}
				\centering
				\includegraphics[scale=.5]{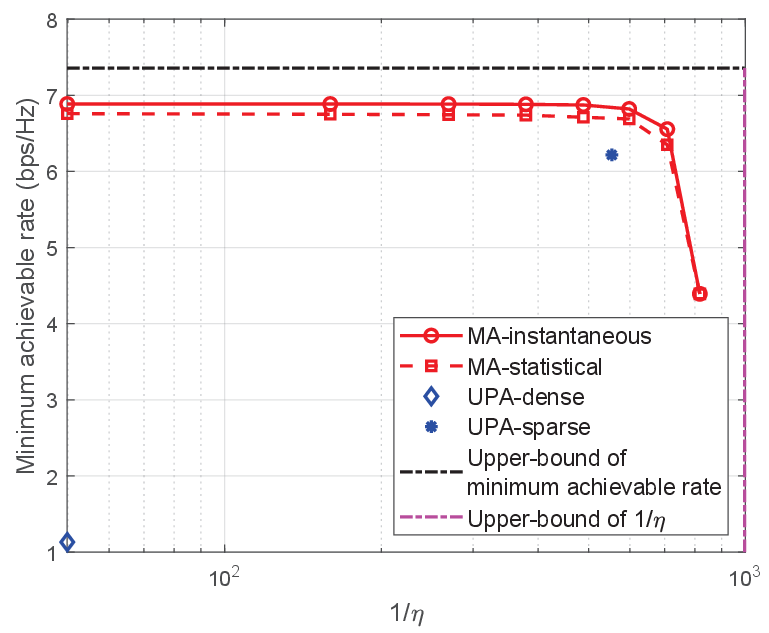}
			\end{minipage}
			\label{region_dense}
		}
		\subfigure[Sparse communication user zone]{
			\begin{minipage}{.47\textwidth}
				\centering
				\includegraphics[scale=.5]{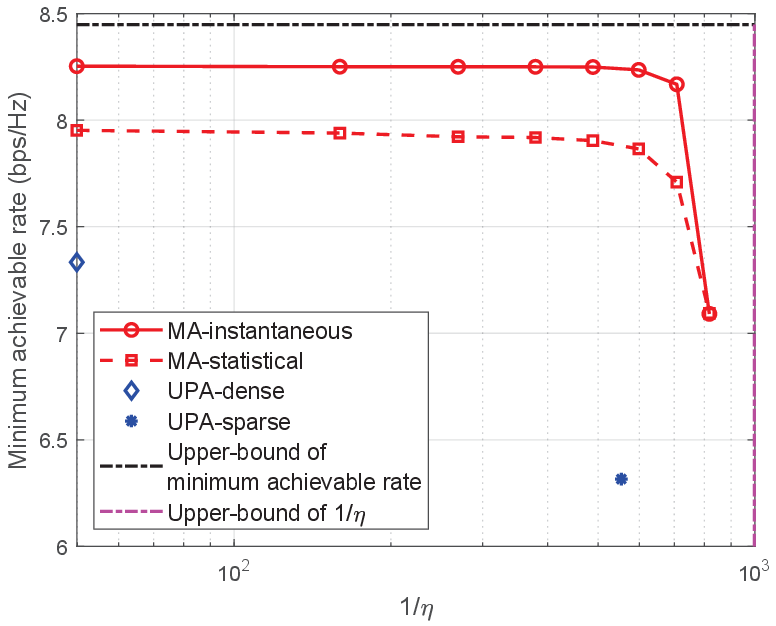}
			\end{minipage}
			\label{region_sparse}
		}
		\caption{Comparison of minimum achievable rate--reciprocal of CRB threshold region for different schemes.}
		\label{FIG4}
	\end{figure}

	\begin{figure*}[!t]
		\centering
		\subfigure[MA-statistical]{
			\begin{minipage}{.31\textwidth}
				\centering
				\includegraphics[scale=.46]{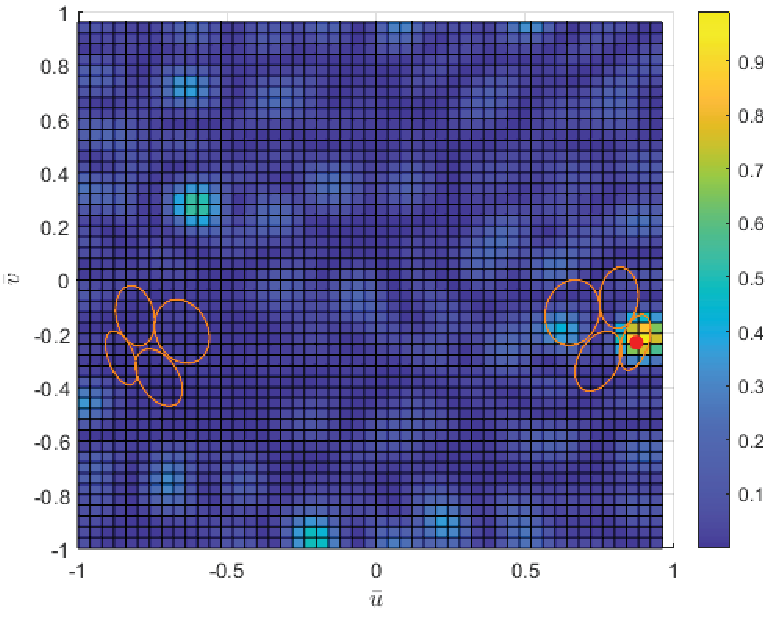}
			\end{minipage}
			\label{correlation_dense_MA}
		}
		\subfigure[UPA-dense]{
			\begin{minipage}{.31\textwidth}
				\centering
				\includegraphics[scale=.46]{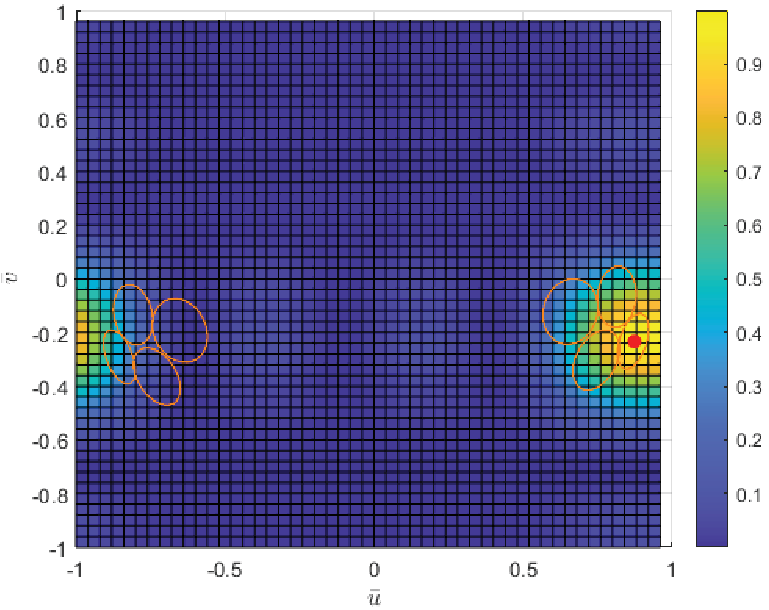}
			\end{minipage}
			\label{correlation_dense_UPA_dense}
		}
		\subfigure[UPA-sparse]{
			\begin{minipage}{.31\textwidth}
				\centering
				\includegraphics[scale=.46]{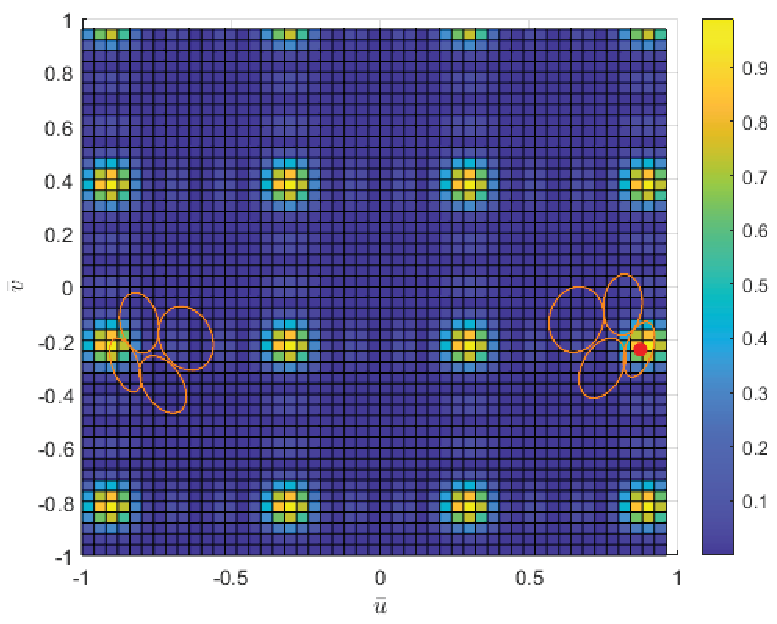}
			\end{minipage}
			\label{correlation_dense_UPA_sparse}
		}
		\caption{Comparison of steering vector correlation by different schemes for dense communication user zone.}
		\label{FIG5}
	\end{figure*}
	
	\begin{figure*}[!t]
		\centering
		\subfigure[MA-statistical]{
			\begin{minipage}{.31\textwidth}
				\centering
				\includegraphics[scale=.46]{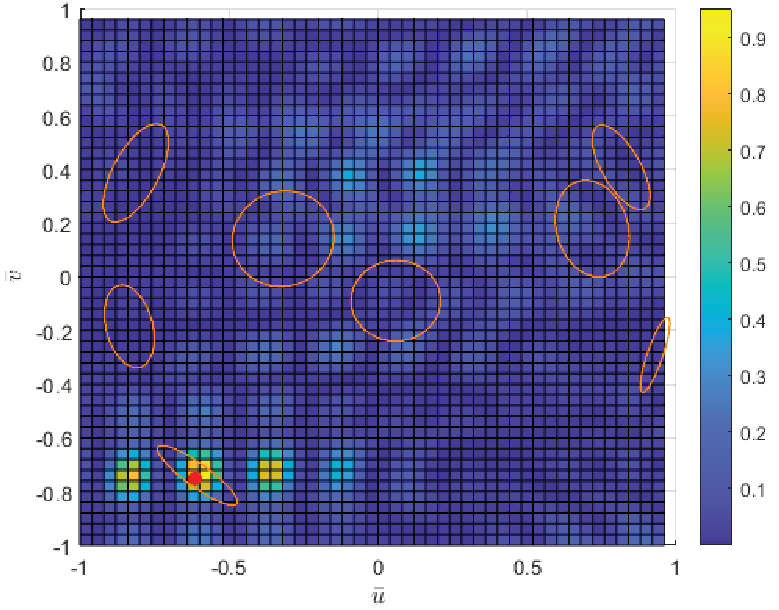}
			\end{minipage}
			\label{correlation_sparse_MA}
		}
		\subfigure[UPA-dense]{
			\begin{minipage}{.31\textwidth}
				\centering
				\includegraphics[scale=.46]{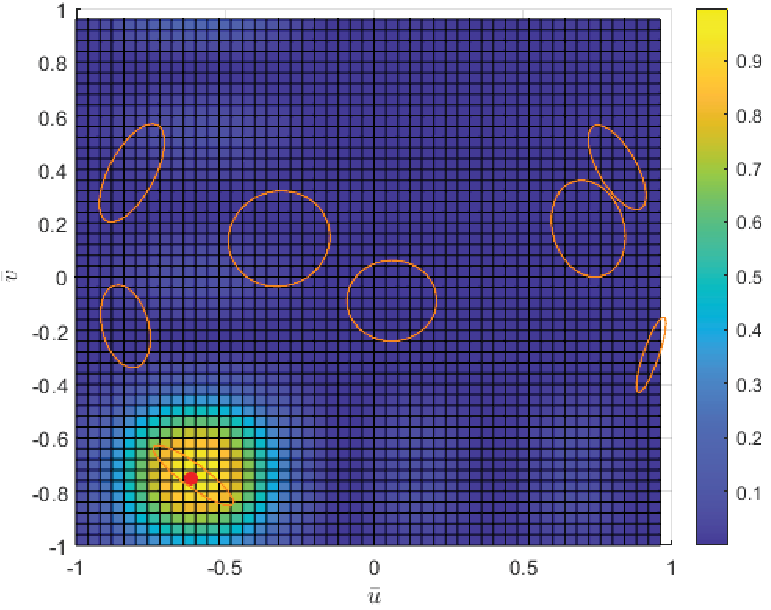}
			\end{minipage}
			\label{correlation_sparse_UPA_dense}
		}
		\subfigure[UPA-sparse]{
			\begin{minipage}{.31\textwidth}
				\centering
				\includegraphics[scale=.46]{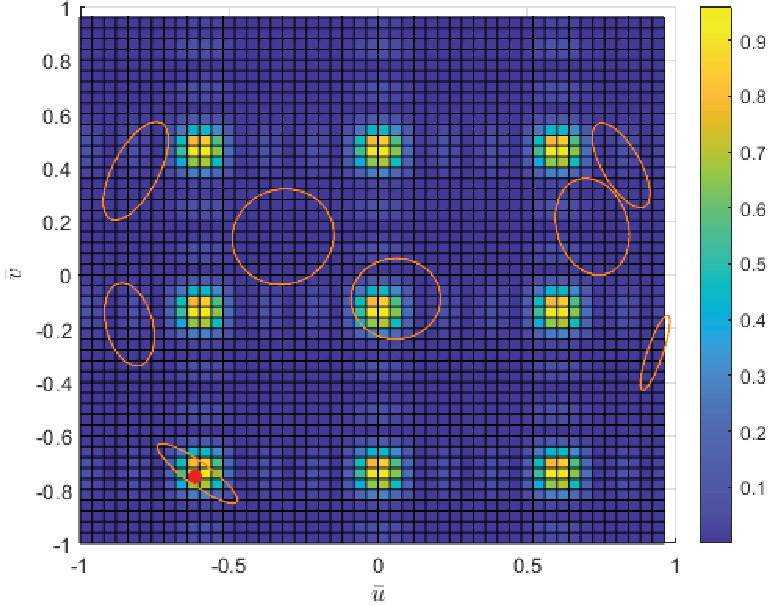}
			\end{minipage}
			\label{correlation_sparse_UPA_sparse}
		}
		\caption{Comparison of steering vector correlation by different schemes for sparse communication user zone.}
		\label{FIG6}
	\end{figure*}

	\begin{figure}[!t]
		\centering
		\subfigure[Dense communication user zone]{
			\begin{minipage}{.47\textwidth}
				\centering
				\includegraphics[scale=.5]{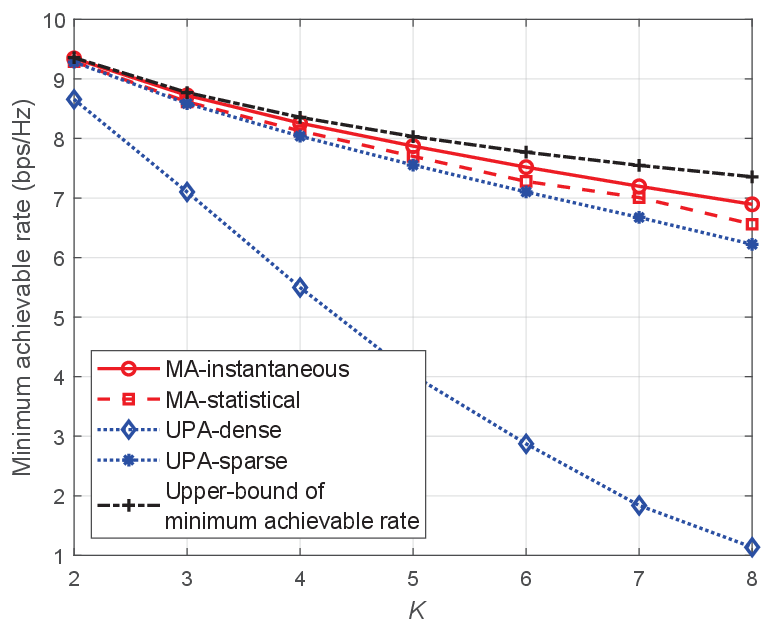}
			\end{minipage}
			\label{K_dense}
		}
		\subfigure[Sparse communication user zone]{
			\begin{minipage}{.47\textwidth}
				\centering
				\includegraphics[scale=.5]{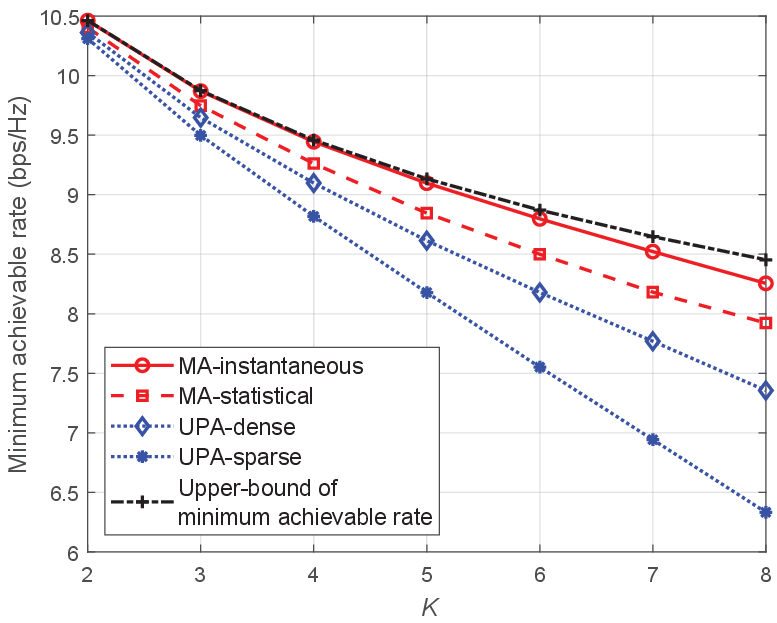}
			\end{minipage}
			\label{K_sparse}
		}
		\caption{Comparison of minimum achievable rate versus $K$ for different schemes.}
		\label{FIG7}
	\end{figure}

	\begin{figure}[!t]
		\centering
		\subfigure[${\rm{MSE}}(u)$]{
			\begin{minipage}{.47\textwidth}
				\centering
				\includegraphics[scale=.5]{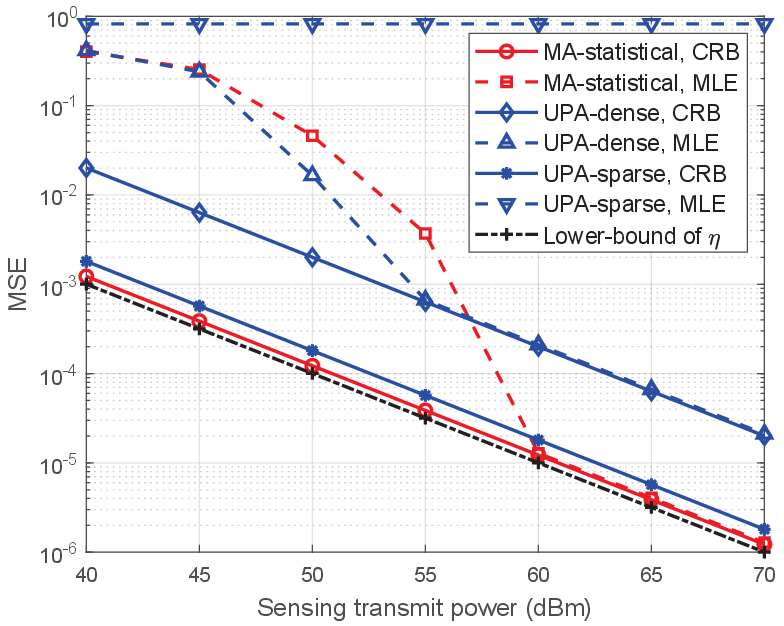}
			\end{minipage}
			\label{MSE_u}
		}
		\subfigure[${\rm{MSE}}(v)$]{
			\begin{minipage}{.47\textwidth}
				\centering
				\includegraphics[scale=.5]{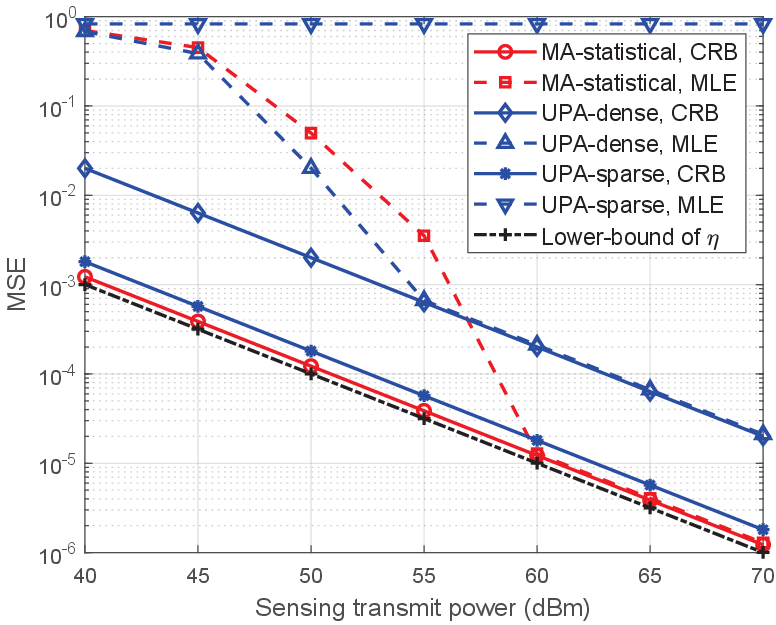}
			\end{minipage}
			\label{MSE_v}
		}
		\caption{Comparison of sensing MSE versus $P^{\rm{s}}$ for different schemes.}
		\label{FIG8}
	\end{figure}

	
	In this section, we present numerical results to evaluate the performance of the proposed MA-aided ISAC system. The antenna moving region is set as a square region of size $A \times A$, with $A = 5\lambda$, $N = 16$, and $\lambda = 0.05$ m. The minimum distance between any two MAs is set as $D_0 = \lambda/2$. For the communication subsystem, the BS's maximum transmit power is set to $P = 20$ dBm, and the noise power is $\sigma^2 = -80$ dBm. We consider $K = 8$ communication users, each randomly distributed within a different 3D sphere of radius $5$ m.	For the sensing subsystem, the transmit power of the probing signal is set to $P^{\rm{s}} = 40$ dBm, and the minimum sensing channel coefficient power is $\tilde{\beta} = 4 \times 10^{-15}$. The target's AoAs are set as $\theta = 45^\circ$ and $\phi = 60^\circ$, leading to $u = \sin \theta \cos \phi = 0.35$ and $v = \cos \theta = 0.71$. The number of snapshots is set to $T = 16$.	In the proposed Algorithm~\ref{alg2}, the convergence thresholds are set as $\epsilon_1 = \epsilon_2 = 10^{-3}$. All numerical results represent the average performance over $Q = 5000$ independent Monte Carlo realizations of users' locations within their respective zones.
	
	
	As illustrated in Fig.~\ref{FIG2}, we consider two typical configurations of communication user zones: a dense communication user zone setup and a sparse communication user zone setup.	In the dense communication user zone setup, all user zones are grouped into two clusters, with the zones within each cluster positioned close to one another. In contrast, in the sparse communication user zone setup, the user zones are spatially well-separated.
	
	To evaluate the performance of the proposed MA scheme based on users' statistical CSI, we consider the following benchmark schemes for comparison: 1) \textbf{MA-instantaneous}: For each realization of users' locations, the positions of MAs are optimized using Algorithm~\ref{alg2} to maximize the instantaneous minimum achievable rate; 2) \textbf{Uniform planar array (UPA)-dense}: The positions $\{\bm{r}_n\}_{n=1}^N$ are configured as a UPA with half-wavelength inter-antenna spacing in both horizontal and vertical directions; 3) \textbf{UPA-sparse}: The positions $\{\bm{r}_n\}_{n=1}^N$ are configured as a UPA with the largest possible aperture, where the inter-antenna spacing is set to $A/\left(\left\lceil\sqrt{N}\right\rceil-1\right)$ in both horizontal and vertical directions. Additionally, we calculate the upper-bound of the expected minimum achievable rate via \eqref{upper} and the lower-bound of the CRB threshold via \eqref{lower}.
	
	First, in Fig.~\ref{convergence}, we show the convergence behavior of the proposed Algorithm~\ref{alg2}. We consider the sparse setup of communication user zone, and set $\eta=0.003$ for target sensing. It is shown that for MA systems based on both users' instantaneous and statistical CSI, the algorithm converges within $10$ iterations, which demonstrates the effectiveness of Algorithm~\ref{alg2}. Moreover, although the MA system based on users' statistical CSI experiences a reduction in the minimum achievable rate among users due to the MAs' positions optimized to cater to the random user locations, it avoids the need for frequent antenna movement required for the MAs' positions optimized based on instantaneous CSI, thereby achieving a favorable balance between communication performance and antenna movement overhead in practical systems.
	
	Next, we show in Fig.~\ref{FIG3} the optimized MAs' positions for different CRB thresholds, $\eta$. We consider the dense setup of communication user zone. For small CRB thresholds, the optimized MAs' positions exhibit symmetry w.r.t. both the $y$-axis and $z$-axis, effectively balancing the estimation accuracy of $u$ and $v$. Moreover, the optimized MAs' positions are as far as possible from the center of the square region to maximize the aperture of the MA array, which enhances angular resolution and improves AoA estimation performance.

	Fig.~\ref{FIG4} compares the minimum achievable rate--reciprocal of CRB threshold region for the proposed and benchmark schemes. It is observed that the MA schemes always outperform other benchmark schemes with FPAs in terms of minimizing sensing CRB	and maximizing the minimum achievable rate. Moreover, the MA-aided ISAC system can achieve flexible trade-offs between sensing and communication performance thanks to	the flexible antenna positioning, whereas the FPA-based ISAC	system cannot alter the communication/sensing performance via antenna positioning. Moreover, new insights can be drawn from the benchmark schemes. When comparing the dense UPA with the sparse UPA, it is shown that increasing the array sparsity reduces the channel correlation among nearby users, thereby enhancing the minimum achievable rate. However, for users located farther away, the increased grating lobes of the sparse UPA can lead to higher channel correlation, ultimately reducing the minimum achievable rate.

	To provide further insights, Fig.~\ref{FIG5} and Fig.~\ref{FIG6} illustrate the steering vector correlation, defined as $\frac{1}{N^2}|\bm{\alpha}(\tilde{\bm{r}},\bar{\bm{n}}^\textrm{zon})^{\mathsf H} \bm{\alpha}(\tilde{\bm{r}},\bar{\bm{\chi}})|^2$, for each scheme versus $\bar{\bm{\chi}} \in [-1,1] \times [-1,1]$. Here, $\bar{\bm{n}}^\textrm{zon} \in \mathbb{R}^{2\times1}$ denotes the 2D wave vector of the center of the first communication user zone.	For the dense setup, the center of the first zone is located at $\bm{n}^\textrm{zon} = [20,41,-11]^{\mathsf T}$ m, resulting in a 2D wave vector of $\bar{\bm{n}}^\textrm{zon} = [\bm{n}^\textrm{zon}[2]/\|\bm{n}^\textrm{zon}\|_2, \bm{n}^\textrm{zon}[3]/\|\bm{n}^\textrm{zon}\|_2]^{\mathsf T} = [0.87, -0.23]^{\mathsf T}$. For the sparse setup, the center of the first zone is located at $\bm{n}^\textrm{zon} = [7,-18,-22]^{\mathsf T}$ m, resulting in a 2D wave vector of $\bar{\bm{n}}^\textrm{zon} = [\bm{n}^\textrm{zon}[2]/\|\bm{n}^\textrm{zon}\|_2, \bm{n}^\textrm{zon}[3]/\|\bm{n}^\textrm{zon}\|_2]^{\mathsf T} = [-0.61, -0.75]^{\mathsf T}$. In Fig.~\ref{FIG5} and Fig.~\ref{FIG6}, the red point represents $\bar{\bm{n}}^\textrm{zon}$, and the orange borders represent the projection of all communication user zones onto the $\bar{u}$-$\bar{v}$ plane. We set $\eta=0.003$ for target sensing. It can be observed in Fig.~\ref{correlation_dense_MA} and Fig.~\ref{correlation_sparse_MA} that the proposed MA scheme has a very narrow main lobe with only a few side-lobes, which can minimize interference leakage to communication users at other locations and significantly enhance multiuser communication. In contrast, the dense UPA scheme has a wide main lobe, leading to higher channel correlation among nearby users in Fig.~\ref{correlation_dense_UPA_dense}. Moreover, the sparse UPA scheme has many grating lobes, leading to higher channel correlation among users located farther away and ultimately reducing the minimum achievable rate in Fig.~\ref{correlation_sparse_UPA_sparse}. Consequently, the MA system outperforms FPA systems by reconfiguring the array geometry dynamically to adapt to different locations of communication user zones.

	Fig.~\ref{FIG7} compares the minimum achievable rate versus $K$ for different schemes. We set $\eta=0.003$ for target sensing. It is observed that the proposed MA schemes always outperform the benchmark schemes with FPAs. For a small number of users, the performance of the MA schemes approaches the upper-bound of the minimum achievable rate. Moreover, the performance gap between the MA and FPA schemes increases with the number of users, highlighting the efficiency of MA systems in enhancing multiuser communication in scenarios with a large number of users.

	Finally, to validate the sensing performance of the designed MA array, Fig.~\ref{FIG8} compares the CRBs and the actual values of ${\rm{MSE}}(u)$ and ${\rm{MSE}}(v)$ versus the sensing transmit power $P^{\rm{s}}$ for different schemes. The MAs' positions are configured as shown in Fig.~\ref{position2}. The results indicate that the curves representing AoA estimation MSE using MLE closely approach the CRB for both the proposed scheme and the UPA-dense scheme with high transmit power. However, the UPA-sparse scheme exhibits multiple grating lobes as shown in Fig.~\ref{correlation_dense_UPA_sparse} and Fig.~\ref{correlation_sparse_UPA_sparse}, leading to ambiguity in distinguishing the true AoA $\bm{\chi} = [0.35, 0.71]^{\mathsf T}$ from false estimates associated with these grating lobes. This ambiguity leads to higher MSE obtained via MLE for the UPA-sparse scheme. Additionally, the proposed scheme achieves significantly lower MSE compared to the benchmark schemes with FPA, with its MSE close to the CRB lower-bound. Furthermore, the similar performance of ${\rm{MSE}}(u)$ and ${\rm{MSE}}(v)$ demonstrates the effectiveness of the designed MA array in accurately estimating both $u$ and $v$.
	

	\section{Conclusions}

	In this paper, we studied ISAC systems aided by the MA array to improve the communication and sensing performance via flexible antenna position optimization. First, we considered the downlink multiuser communication, where each user is randomly distributed  within a given 3D zone with local movement. To reduce the overhead of frequent antenna movement, the APV was designed based on users' statistical CSI, so that the antenna movement is implemented in large timescales. Then, for target sensing, the CRBs of the estimation MSE for different spatial AoAs were derived as functions of MAs' positions. Based on the above, we formulated an optimization problem to maximize the expected minimum achievable rate among all communication users, with given constraints on the maximum acceptable CRB thresholds for target sensing. An alternating optimization algorithm was proposed to iteratively optimize one of the horizontal and vertical APVs of the MA array with the other being fixed. Numerical results demonstrated that our proposed MA arrays can significantly enlarge the trade-off region between communication and sensing performance compared to conventional FPA arrays with different inter-antenna spacing. It was also revealed that the steering vectors of the designed MA arrays exhibit low correlation in the angular domain, thus effectively reducing channel correlation among communication users to enhance their achievable rates, while alleviating ambiguity in target angle estimation to achieve improved sensing accuracy.

	\appendix
	\subsection{Proof of Theorem 1}
	 By vectorizing \eqref{Y}, we have
	 \begin{align}
	 	\bm{y}^{\rm{s}}&\triangleq{\rm{vec}}(\bm{Y}^{\rm{s}}) \\ 
	 	&= {\rm{vec}}\left(\bm{H}^{\rm{s}}(\tilde{\bm{r}},\bm{\chi})\bm{S}^{\rm{s}}\right) + {\rm{vec}}(\bm{Z}^{\rm{s}}) \notag\\
	 	&= \beta{\rm{vec}}\left(\bm{\alpha}(\tilde{\bm{r}},\bm{\chi})\bm{\alpha}(\tilde{\bm{r}},\bm{\chi})^{\mathsf T} \bm{S}^{\rm{s}}\right) + {\rm{vec}}(\bm{Z}^{\rm{s}}) \notag \\
	 	&\triangleq \beta\bm{b}(\tilde{\bm{r}},\bm{\chi}) + \bm{z}^{\rm{s}}, \notag
	 \end{align}
 	where $\bm{y}^{\rm{s}}\triangleq{\rm{vec}}(\bm{Y}^{\rm{s}}) \in\mathbb{C}^{NT\times1}$, $\bm{b}(\tilde{\bm{r}},\bm{\chi})\triangleq {\rm{vec}}\left(\bm{\alpha}(\tilde{\bm{r}},\bm{\chi})\bm{\alpha}(\tilde{\bm{r}},\bm{\chi})^{\mathsf T} \bm{S}^{\rm{s}}\right) \in\mathbb{C}^{NT\times1}$, and $\bm{z}^{\rm{s}}\triangleq{\rm{vec}}(\bm{Z}^{\rm{s}}) \in\mathbb{C}^{NT\times1}\sim \mathcal{CN}(0,\bm{R}_{\bm{z}})$, with $\bm{R}_{\bm{z}}=\sigma^2\bm{I}_{NT}$. Then, the MLE of the unknown parameters $\beta$ and $\bm{\chi}$ can be written as
 	\begin{align}\label{hatbetaeta}
 		\left(\hat{\beta},\hat{\bm{\chi}}\right) = \arg\min_{\bar{\beta},\bar{\bm{\chi}}} \|\bm{y}^{\rm{s}} - \bar{\beta}\bm{b}(\tilde{\bm{r}},\bar{\bm{\chi}})\|_2^2.
 	\end{align}
 	With any given $\bm{\chi}$, the optimal estimation of $\beta$ is given by
 	\begin{align}\label{hatbeta}
 		\hat{\beta} = \frac{\bm{b}(\tilde{\bm{r}},\bm{\chi})^{\mathsf H} \bm{y}^{\rm{s}}}{\|\bm{b}(\tilde{\bm{r}},\bm{\chi})\|_2^2}.
 	\end{align}
 	Then, by substituting \eqref{hatbeta} back into \eqref{hatbetaeta}, we have
 	\begin{align}
 		&\|\bm{y}^{\rm{s}} - \hat{\beta}\bm{b}(\tilde{\bm{r}},\bm{\chi})\|_2^2 \\
 		=& \left(\bm{y}^{\rm{s}} - \hat{\beta}\bm{b}(\tilde{\bm{r}},\bm{\chi})\right)^{\mathsf H} \left(\bm{y}^{\rm{s}} - \hat{\beta}\bm{b}(\tilde{\bm{r}},\bm{\chi})\right) \notag\\
 		=&\|\bm{y}^{\rm{s}}\|_2^2 + |\hat{\beta}|^2 \|\bm{b}(\tilde{\bm{r}},\bm{\chi})\|_2^2  - 2\Re\left\{\hat{\beta}\left(\bm{y}^{\rm{s}}\right)^{\mathsf H} \bm{b}(\tilde{\bm{r}},\bm{\chi})\right\} \notag\\
 		=&\|\bm{y}^{\rm{s}}\|_2^2 - \frac{\left|(\bm{y}^{\rm{s}})^{\mathsf H} \bm{b}(\tilde{\bm{r}},\bm{\chi})\right|^2}{\|\bm{b}(\tilde{\bm{r}},\bm{\chi})\|_2^2}. \notag
 	\end{align}
 	On one hand, the numerator of $\frac{\left|(\bm{y}^{\rm{s}})^{\mathsf H} \bm{b}(\tilde{\bm{r}},\bm{\chi})\right|^2}{\|\bm{b}(\tilde{\bm{r}},\bm{\chi})\|_2^2}$ can be further written as
 	\begin{align}
 		\left|(\bm{y}^{\rm{s}})^{\mathsf H} \bm{b}(\tilde{\bm{r}},\bm{\chi})\right|^2 =& \left|{\rm{vec}}(\bm{Y}^{\rm{s}})^{\mathsf H} {\rm{vec}}\left(\bm{\alpha}(\tilde{\bm{r}},\bm{\chi})\bm{\alpha}(\tilde{\bm{r}},\bm{\chi})^{\mathsf T} \bm{S}^{\rm{s}}\right) \right|^2 \notag\\
 		\overset{(a_1)}=& \left|{\rm{Tr}}\left((\bm{Y}^{\rm{s}})^{\mathsf H} \bm{\alpha}(\tilde{\bm{r}},\bm{\chi})\bm{\alpha}(\tilde{\bm{r}},\bm{\chi})^{\mathsf T} \bm{S}^{\rm{s}}\right) \right|^2 \\
 		\overset{(a_2)}=& \left|{\rm{Tr}}\left(\bm{\alpha}(\tilde{\bm{r}},\bm{\chi})^{\mathsf T} \bm{S}^{\rm{s}} (\bm{Y}^{\rm{s}})^{\mathsf H} \bm{\alpha}(\tilde{\bm{r}},\bm{\chi}) \right) \right|^2 \notag\\
 		\overset{(a_3)}=& \left|{\rm{vec}}\left(\bm{\alpha}(\tilde{\bm{r}},\bm{\chi})^{\mathsf T} \bm{S}^{\rm{s}} (\bm{Y}^{\rm{s}})^{\mathsf H} \bm{\alpha}(\tilde{\bm{r}},\bm{\chi}) \right) \right|^2 \notag\\
 		\overset{(a_4)}=& \left|\left(\bm{\alpha}(\tilde{\bm{r}},\bm{\chi})\otimes \bm{\alpha}(\tilde{\bm{r}},\bm{\chi})\right)^{\mathsf T} {\rm{vec}}\left(\bm{S}^{\rm{s}}(\bm{Y}^{\rm{s}})^{\mathsf H}\right)\right|^2, \notag
 	\end{align}
 	where the equality $(a_1)$ holds because ${\rm{vec}}(\bm{A})^{\mathsf H} {\rm{vec}}(\bm{B}) = {\rm{Tr}}(\bm{A}^{\mathsf H} \bm{B})$ for $\bm{A}$ and $\bm{B}$ of equal size. The equality $(a_2)$ holds because ${\rm{Tr}}(\bm{A}\bm{B}) = {\rm{Tr}}(\bm{B}\bm{A})$ for $\bm{A} \in{\mathbb{C}^{p \times q}}$ and $\bm{B} \in{\mathbb{C}^{q \times p}}$. The equality $(a_3)$ holds because $\bm{\alpha}(\tilde{\bm{r}},\bm{\chi})^{\mathsf T} \bm{S}^{\rm{s}} (\bm{Y}^{\rm{s}})^{\mathsf H} \bm{\alpha}(\tilde{\bm{r}},\bm{\chi})$ is a scalar. The equality $(a_4)$ holds because ${\rm{vec}}(\bm{A}\bm{B}\bm{C}) = (\bm{C}^{\mathsf T} \otimes \bm{A}) {\rm{vec}}(\bm{B})$ for $\bm{A} \in{\mathbb{C}^{p \times q}}$, $\bm{B} \in{\mathbb{C}^{q \times r}}$, and $\bm{C} \in{\mathbb{C}^{r \times s}}$.
 	
 	On the other hand, the denominator of $\frac{\left|(\bm{y}^{\rm{s}})^{\mathsf H} \bm{b}(\tilde{\bm{r}},\bm{\chi})\right|^2}{\|\bm{b}(\tilde{\bm{r}},\bm{\chi})\|_2^2}$ can be further written as
 	\begin{align}
 		&\|\bm{b}(\tilde{\bm{r}},\bm{\chi})\|_2^2 \\
 		&= {\rm{vec}}\left(\bm{\alpha}(\tilde{\bm{r}},\bm{\chi})\bm{\alpha}(\tilde{\bm{r}},\bm{\chi})^{\mathsf T} \bm{S}^{\rm{s}}\right)^{\mathsf H} {\rm{vec}}\left(\bm{\alpha}(\tilde{\bm{r}},\bm{\chi})\bm{\alpha}(\tilde{\bm{r}},\bm{\chi})^{\mathsf T} \bm{S}^{\rm{s}}\right) \notag\\
 		&={\rm{Tr}}\left( (\bm{S}^{\rm{s}})^{\mathsf H}\bm{\alpha}(\tilde{\bm{r}},\bm{\chi})^{\mathsf *}\bm{\alpha}(\tilde{\bm{r}},\bm{\chi})^{\mathsf H} \bm{\alpha}(\tilde{\bm{r}},\bm{\chi})\bm{\alpha}(\tilde{\bm{r}},\bm{\chi})^{\mathsf T} \bm{S}^{\rm{s}} \right) \notag\\
 		&\overset{(b_1)}=N{\rm{Tr}}\left(\bm{\alpha}(\tilde{\bm{r}},\bm{\chi})^{\mathsf T} \bm{S}^{\rm{s}}  (\bm{S}^{\rm{s}})^{\mathsf H}\bm{\alpha}(\tilde{\bm{r}},\bm{\chi})^{\mathsf *} \right) \notag\\
 		&\overset{(b_2)}=P^{\rm{s}}TN, \notag
 	\end{align}
 	where the equality $(b_1)$ holds because $\bm{\alpha}(\tilde{\bm{r}},\bm{\chi})^{\mathsf H} \bm{\alpha}(\tilde{\bm{r}},\bm{\chi}) = N$ and ${\rm{Tr}}(\bm{A}\bm{B}) = {\rm{Tr}}(\bm{B}\bm{A})$ for $\bm{A} \in{\mathbb{C}^{p \times q}}$ and $\bm{B} \in{\mathbb{C}^{q \times p}}$. The equality $(b_2)$ holds because $\bm{S}^{\rm{s}}(\bm{S}^{\rm{s}})^{\mathsf H} = \frac{P^{\rm{s}}T}{N} \bm{I}_N$.
 	
 	Since both $\|\bm{y}^{\rm{s}}\|_2^2$ and $\|\bm{b}(\tilde{\bm{r}},\bm{\chi})\|_2^2$ are constant w.r.t. $\bm{\chi}$, the MLE of $\bm{\chi}$ is given by
 	\begin{align}
 		\hat{\bm{\chi}} &= \arg\min_{\bar{\bm{\chi}}} \|\bm{y}^{\rm{s}}\|_2^2 - \frac{\left|(\bm{y}^{\rm{s}})^{\mathsf H} \bm{b}(\tilde{\bm{r}},\bar{\bm{\chi}})\right|^2}{\|\bm{b}(\tilde{\bm{r}},\bar{\bm{\chi}})\|_2^2} \\
 		&= \arg\max_{\bar{\bm{\chi}}} \left|\left(\bm{\alpha}(\tilde{\bm{r}},\bar{\bm{\chi}})\otimes \bm{\alpha}(\tilde{\bm{r}},\bar{\bm{\chi}})\right)^{\mathsf T} {\rm{vec}}\left(\bm{S}^{\rm{s}}(\bm{Y}^{\rm{s}})^{\mathsf H}\right)\right|^2. \notag
 	\end{align}
 	This thus completes the proof of Theorem 1. 
	
	\subsection{Proof of Theorem 2}
	Let $\bm{\kappa}\triangleq[u,v,\Re(\beta),\Im(\beta)]^{\mathsf T} \in\mathbb{R}^{4\times1}$ denote the vector of unknown parameters
	to be estimated, which includes the two spatial AoAs and the real and imaginary parts of complex	channel coefficient. To facilitate deriving the CRB expression, we define $\bm{f}(\bm{\chi},\bm{\beta}) \triangleq \beta\bm{b}(\tilde{\bm{r}},\bm{\chi})\in \mathbb{C}^{NT\times1}$ with $\bm{\beta}\triangleq[\Re(\beta),\Im(\beta)]^{\mathsf T}$. Let $\bm{F} \in\mathbb{R}^{4\times 4}$ denote the Fisher information matrix (FIM) for estimating $\bm{\kappa}$. Then, the entry in the $p$th ($p=1,2,\ldots,4$) row and $q$th ($q=1,2,\ldots,4$) column of $\bm{F}$ can be written as 
	\begin{align}
		\bm{F}[p,q] &= 2\Re\left( \frac{\partial \bm{f}(\bm{\chi},\bm{\beta})^{\mathsf H}}{\partial \bm{\kappa}[p]} \bm{R}_{\bm{z}}^{-1} \frac{\partial \bm{f}(\bm{\chi},\bm{\beta})}{\partial \bm{\kappa}[q]} \right) \notag\\
		&=\frac{2}{\sigma^2}\Re\left( \frac{\partial \bm{f}(\bm{\chi},\bm{\beta})^{\mathsf H}}{\partial \bm{\kappa}[p]} \frac{\partial \bm{f}(\bm{\chi},\bm{\beta})}{\partial \bm{\kappa}[q]} \right).
	\end{align}
	Accordingly, the FIM $\bm{F}$ can be partitioned as
	\begin{align}\label{F}
		\bm{F} = \begin{bmatrix}
			\bm{J}_{\bm{\chi},\bm{\chi}}  & \bm{J}_{\bm{\chi},\bm{\beta}} \\
			\bm{J}_{\bm{\chi},\bm{\beta}}^{\mathsf T}  & \bm{J}_{\bm{\beta},\bm{\beta}} 
		\end{bmatrix},
	\end{align}
	where $\bm{J}_{\bm{\chi},\bm{\chi}}\in \mathbb{R}^{2 \times 2}$, $\bm{J}_{\bm{\chi},\bm{\beta}}\in \mathbb{R}^{2 \times 2}$, and $\bm{J}_{\bm{\beta},\bm{\beta}}\in \mathbb{R}^{2 \times 2}$ are given by
	\begin{align}\label{J}
		\bm{J}_{\bm{\chi},\bm{\chi}} &= \frac{2}{\sigma^2}\Re\left( \frac{\partial \bm{f}(\bm{\chi},\bm{\beta})^{\mathsf H}}{\partial \bm{\chi}} \frac{\partial \bm{f}(\bm{\chi},\bm{\beta})}{\partial \bm{\chi}} \right) \notag\\
		\bm{J}_{\bm{\chi},\bm{\beta}} &= \frac{2}{\sigma^2}\Re\left( \frac{\partial \bm{f}(\bm{\chi},\bm{\beta})^{\mathsf H}}{\partial \bm{\chi}} \frac{\partial \bm{f}(\bm{\chi},\bm{\beta})}{\partial \bm{\beta}} \right) \notag\\
		\bm{J}_{\bm{\beta},\bm{\beta}} &= \frac{2}{\sigma^2}\Re\left( \frac{\partial \bm{f}(\bm{\chi},\bm{\beta})^{\mathsf H}}{\partial \bm{\beta}} \frac{\partial \bm{f}(\bm{\chi},\bm{\beta})}{\partial \bm{\beta}} \right).
	\end{align}
	Then, based on the inverse formula of the block matrix, the sub-matrix formed by the first two rows and first two columns of the inverse matrix of $\bm{F}$ can be written as
	\begin{align}\label{Lambda}
		\bm{\Lambda}\triangleq \begin{bmatrix}
			\bm{F}^{-1}[1,1]  & \bm{F}^{-1}[1,2] \\
			\bm{F}^{-1}[2,1]  & \bm{F}^{-1}[2,2] 
		\end{bmatrix} = \left[\bm{J}_{\bm{\chi},\bm{\chi}} - \bm{J}_{\bm{\chi},\bm{\beta}} \bm{J}_{\bm{\beta},\bm{\beta}}^{-1} \bm{J}_{\bm{\chi},\bm{\beta}}^{\mathsf T}\right]^{-1}.
	\end{align}
	Then, we can derive the CRB expression of $\bm{\chi}$ as
	\begin{align}
		{\rm{CRB}}_u(\tilde{\bm{r}})=\bm{\Lambda}[1,1],\notag\\
		{\rm{CRB}}_v(\tilde{\bm{r}})=\bm{\Lambda}[2,2].
	\end{align}
	For ease of notation, we simply re-denote $\bm{\alpha}(\tilde{\bm{r}},\bm{\chi})$ as $\bm{\alpha}$ in the following, and define
	\begin{align}
		\dot{\bm{\alpha}}_{\bm{\chi}} \triangleq \frac{\partial \bm{\alpha}}{\partial \bm{\chi}} = \left[\dot{\bm{\alpha}}_u, \dot{\bm{\alpha}}_v\right] \in \mathbb{C}^{N\times 2},
	\end{align}
	where
	\begin{align}
		\dot{\bm{\alpha}}_u &\triangleq \frac{\partial \bm{\alpha}}{\partial u} = j\frac{2\pi}{\lambda} \bm{D}_y\bm{\alpha} \in \mathbb{C}^{N\times1}, \notag\\
		\dot{\bm{\alpha}}_v &\triangleq \frac{\partial \bm{\alpha}}{\partial v} = j\frac{2\pi}{\lambda} \bm{D}_z\bm{\alpha} \in \mathbb{C}^{N\times1},
	\end{align}
	with $\bm{D}_y\triangleq\textrm{diag}(\bm{y})$ and $\bm{D}_z\triangleq\textrm{diag}(\bm{z})$. Then, we have
	\begin{align}
		\frac{\partial \bm{f}(\bm{\chi},\bm{\beta})}{\partial \bm{\chi}} &= \beta\big[{\rm{vec}}\left(\left(\dot{\bm{\alpha}}_u\bm{\alpha}^{\mathsf T} + \bm{\alpha}\dot{\bm{\alpha}}_u^{\mathsf T}\right) \bm{S}^{\rm{s}}\right), \notag\\
		&~~~~ {\rm{vec}}\left(\left(\dot{\bm{\alpha}}_v\bm{\alpha}^{\mathsf T} + \bm{\alpha}\dot{\bm{\alpha}}_v^{\mathsf T}\right) \bm{S}^{\rm{s}}\right) \big] \notag\\
		&\triangleq \beta\left[{\rm{vec}}\left(\dot{\bm{A}}_u \bm{S}^{\rm{s}}\right),  {\rm{vec}}\left(\dot{\bm{A}}_v \bm{S}^{\rm{s}}\right) \right],
	\end{align}
	where $\dot{\bm{A}}_u\triangleq \dot{\bm{\alpha}}_u\bm{\alpha}^{\mathsf T} + \bm{\alpha}\dot{\bm{\alpha}}_u^{\mathsf T} \in\mathbb{C}^{N \times N}$ and $\dot{\bm{A}}_v\triangleq \dot{\bm{\alpha}}_v\bm{\alpha}^{\mathsf T} + \bm{\alpha}\dot{\bm{\alpha}}_v^{\mathsf T} \in\mathbb{C}^{N \times N}$.	Moreover, we have
	\begin{align}\label{fbeta}
		\frac{\partial \bm{f}(\bm{\chi},\bm{\beta})}{\partial \bm{\beta}} = \frac{\partial \beta\bm{b}(\tilde{\bm{r}},\bm{\chi})}{\partial \bm{\beta}} = \bm{b}(\tilde{\bm{r}},\bm{\chi}) [1,j].
	\end{align}
	Then, $\bm{J}_{\bm{\chi},\bm{\chi}}$ in \eqref{J} can be further written as
	\begin{align}
		\bm{J}_{\bm{\chi},\bm{\chi}} &\overset{(c)}= \frac{2}{\sigma^2}\Re\left( \left(\frac{\partial \bm{f}(\bm{\chi},\bm{\beta})}{\partial \bm{\chi}}\right)^{\mathsf H} \frac{\partial \bm{f}(\bm{\chi},\bm{\beta})}{\partial \bm{\chi}} \right) \notag\\
		&\triangleq \frac{2|\beta|^2}{\sigma^2}\Re\left( \begin{bmatrix}
			Q_{u,u} & Q_{u,v} \\
			Q_{v,u}  & Q_{v,v}
		\end{bmatrix} \right),
	\end{align}
	where the equality  $(c)$ holds since $\frac{\partial \bm{f}(\bm{\chi},\bm{\beta})^{\mathsf H}}{\partial \bm{\chi}} = \left(\frac{\partial \bm{f}(\bm{\chi},\bm{\beta})}{\partial \bm{\chi}}\right)^{\mathsf H}$. $Q_{a,b}\triangleq{\rm{vec}}(\dot{\bm{A}}_a \bm{S}^{\rm{s}})^{\mathsf H} {\rm{vec}}(\dot{\bm{A}}_b \bm{S}^{\rm{s}})$, $a,b\in\{u,v\}$, which can be further represented as
	\begin{align}\label{Qab}
		Q_{a,b} &\overset{(d_1)}= {\rm{Tr}} \left(\dot{\bm{A}}_b \bm{S}^{\rm{s}} (\bm{S}^{\rm{s}})^{\mathsf H}(\dot{\bm{A}}_a)^{\mathsf H}\right) \\
		&\overset{(d_2)}=\frac{P^{\rm{s}}T}{N}{\rm{Tr}} \left(\dot{\bm{A}}_b (\dot{\bm{A}}_a)^{\mathsf H}\right) \notag\\
		&= \frac{P^{\rm{s}}T}{N} \Big({\rm{Tr}}\left( \dot{\bm{\alpha}}_b\bm{\alpha}^{\mathsf T} \bm{\alpha}^{\mathsf *} \dot{\bm{\alpha}}_a^{\mathsf H} \right) + {\rm{Tr}}\left( \dot{\bm{\alpha}}_b\bm{\alpha}^{\mathsf T}  \dot{\bm{\alpha}}_a^{\mathsf *}\bm{\alpha}^{\mathsf H}  \right) \notag\\
		&~~~~ + {\rm{Tr}}\left( \bm{\alpha}\dot{\bm{\alpha}}_b^{\mathsf T} \bm{\alpha}^{\mathsf *} \dot{\bm{\alpha}}_a^{\mathsf H} \right) + {\rm{Tr}}\left( \bm{\alpha}\dot{\bm{\alpha}}_b^{\mathsf T}  \dot{\bm{\alpha}}_a^{\mathsf *}\bm{\alpha}^{\mathsf H}  \right)\Big) \notag\\
		&\overset{(d_3)}= \frac{P^{\rm{s}}T}{N} \big( N\dot{\bm{\alpha}}_a^{\mathsf H} \dot{\bm{\alpha}}_b - j\frac{2\pi}{\lambda} \gamma_{c(a)} \bm{\alpha}^{\mathsf H}\dot{\bm{\alpha}}_b \notag\\
		&~~~~ + j\frac{2\pi}{\lambda} \gamma_{c(b)} \dot{\bm{\alpha}}_a^{\mathsf H} \bm{\alpha} + N\dot{\bm{\alpha}}_b^{\mathsf T}\dot{\bm{\alpha}}_a^{\mathsf *} \big) \notag\\
		&\overset{(d_4)}= \frac{8\pi^2P^{\rm{s}}T}{\lambda^2N} \left( N\bar{\gamma}_{c(a),c(b)} + \gamma_{c(a)}\gamma_{c(b)} \right), \notag
	\end{align}
	where the equality $(d_1)$ holds because ${\rm{vec}}(\bm{A})^{\mathsf H} {\rm{vec}}(\bm{B}) = {\rm{Tr}}(\bm{B}\bm{A}^{\mathsf H})$ for $\bm{A}$ and $\bm{B}$ of equal size. The equality $(d_2)$ holds because $\bm{S}^{\rm{s}}(\bm{S}^{\rm{s}})^{\mathsf H} = \frac{P^{\rm{s}}T}{N} \bm{I}_N$. The equality $(d_3)$ holds because $\bm{\alpha}^{\mathsf H} \bm{\alpha} = N$; ${\rm{Tr}}(\bm{a}\bm{b}^{\mathsf H}) = \bm{b}^{\mathsf H}\bm{a}$ for $\bm{a}$ and $\bm{b}$ of equal size; $c(a)\triangleq\begin{cases}
		y,&{a=u}\\
		z,&{a=v}
	\end{cases}$; and $\bm{\alpha}^{\mathsf H}  \dot{\bm{\alpha}}_a = j\frac{2\pi}{\lambda}\sum_{n=1}^{N}\bm{D}_{c(a)}[n,n] = j\frac{2\pi}{\lambda}\sum_{n=1}^{N} c(a)_n \triangleq j\frac{2\pi}{\lambda}\gamma_{c(a)}$. The equality $(d_4)$ holds because $\dot{\bm{\alpha}}_a^{\mathsf H} \dot{\bm{\alpha}}_b = \frac{4\pi^2}{\lambda^2}\sum_{n=1}^{N}\bm{D}_{c(a)}[n,n]\bm{D}_{c(b)}[n,n] = \frac{4\pi^2}{\lambda^2}\sum_{n=1}^{N} c(a)_n c(b)_n \triangleq \frac{4\pi^2}{\lambda^2}\bar{\gamma}_{c(a),c(b)}$. Then, $\bm{J}_{\bm{\chi},\bm{\chi}}$ can be simplified as
	\begin{align}
		\bm{J}_{\bm{\chi},\bm{\chi}} &= \frac{16\pi^2P^{\rm{s}}T|\beta|^2}{\sigma^2\lambda^2N}\begin{bmatrix}
			N\bar{\gamma}_{y,y} + \gamma_{y}\gamma_{y} & N\bar{\gamma}_{y,z} + \gamma_{y}\gamma_{z} \\
			N\bar{\gamma}_{y,z} + \gamma_{y}\gamma_{z} & N\bar{\gamma}_{z,z} + \gamma_{z}\gamma_{z}
		\end{bmatrix}.
	\end{align}
	Moreover, $\bm{J}_{\bm{\chi},\bm{\beta}}$ can be further written as
	\begin{align}
		\bm{J}_{\bm{\chi},\bm{\beta}} &= \frac{2}{\sigma^2}\Re\left( \frac{\partial \bm{f}(\bm{\chi},\bm{\beta})^{\mathsf H}}{\partial \bm{\chi}} \frac{\partial \bm{f}(\bm{\chi},\bm{\beta})}{\partial \bm{\beta}} \right) \notag\\
		&= \frac{2}{\sigma^2}\Re\left(\begin{bmatrix}
			\beta^{\mathsf *}{\rm{vec}}(\dot{\bm{A}}_u \bm{S}^{\rm{s}})^{\mathsf H} \bm{b}(\tilde{\bm{r}},\bm{\chi}) [1,j] \\
			\beta^{\mathsf *}{\rm{vec}}(\dot{\bm{A}}_v \bm{S}^{\rm{s}})^{\mathsf H} \bm{b}(\tilde{\bm{r}},\bm{\chi}) [1,j]
		\end{bmatrix}  \right),
	\end{align}
	where ${\rm{vec}}(\dot{\bm{A}}_u \bm{S}^{\rm{s}})^{\mathsf H} \bm{b}(\tilde{\bm{r}},\bm{\chi})$ can be derived in a manner similar to the procedure in \eqref{Qab} as
	\begin{align}
		&{\rm{vec}}(\dot{\bm{A}}_u \bm{S}^{\rm{s}})^{\mathsf H} \bm{b}(\tilde{\bm{r}},\bm{\chi}) \\
		=& {\rm{Tr}}\left(\bm{\alpha}\bm{\alpha}^{\mathsf T} \bm{S}^{\rm{s}} (\bm{S}^{\rm{s}})^{\mathsf{H}}(\dot{\bm{A}}_u)^{\mathsf H} \right)\notag\\
		=& \frac{P^{\rm{s}}T}{N} \left({\rm{Tr}}\left( \bm{\alpha}\bm{\alpha}^{\mathsf T} \bm{\alpha}^{\mathsf *} \dot{\bm{\alpha}}_u^{\mathsf H} \right) + {\rm{Tr}}\left( \bm{\alpha}\bm{\alpha}^{\mathsf T}  \dot{\bm{\alpha}}_u^{\mathsf *}\bm{\alpha}^{\mathsf H}  \right) \right) \notag\\
		=& \frac{P^{\rm{s}}T}{N} \left( N\dot{\bm{\alpha}}_u^{\mathsf H}\bm{\alpha}  -  j\frac{2\pi}{\lambda}\gamma_{y}\bm{\alpha}^{\mathsf H}\bm{\alpha}  \right) \notag\\
		=& -j\frac{4\pi P^{\rm{s}}T\gamma_{y}}{\lambda}. \notag
	\end{align}
	Similarly, ${\rm{vec}}(\dot{\bm{A}}_v \bm{S}^{\rm{s}})^{\mathsf H} \bm{b}(\tilde{\bm{r}},\bm{\chi})$ can be further written as
	\begin{align}
		{\rm{vec}}(\dot{\bm{A}}_v \bm{S}^{\rm{s}})^{\mathsf H} \bm{b}(\tilde{\bm{r}},\bm{\chi}) = -j\frac{4\pi P^{\rm{s}}T\gamma_{z}}{\lambda}.
	\end{align}
	Then, $\bm{J}_{\bm{\chi},\bm{\beta}}$ can be simplified as
	\begin{align}
		\bm{J}_{\bm{\chi},\bm{\beta}} &= \frac{8\pi P^{\rm{s}}T}{\sigma^2\lambda}\Re\left(-j\beta^{\mathsf *}\begin{bmatrix}
			\gamma_{y} & j\gamma_{y} \\
			\gamma_{z} & j\gamma_{z}
		\end{bmatrix}  \right) \\
		&=  \frac{8\pi P^{\rm{s}}T}{\sigma^2\lambda}\begin{bmatrix}
			-\Im(\beta)\gamma_{y} & \Re(\beta)\gamma_{y} \\
			-\Im(\beta)\gamma_{z} & \Re(\beta)\gamma_{z}
		\end{bmatrix}  \notag\\
		&\triangleq  \frac{8\pi P^{\rm{s}}T}{\sigma^2\lambda}\bm{\gamma}\bar{\bm{\beta}}^{\mathsf T},  \notag
	\end{align}
	where $\bm{\gamma}\triangleq[\gamma_{y},\gamma_{z}]^{\mathsf T}$ and $\bar{\bm{\beta}}\triangleq[-\Im(\beta),\Re(\beta)]^{\mathsf T}$.
	
	Next, substituting \eqref{fbeta} into \eqref{J}, $\bm{J}_{\bm{\beta},\bm{\beta}}$ can be derived in a manner similar to the procedure in \eqref{Qab} as
	\begin{align}
		\bm{J}_{\bm{\beta},\bm{\beta}} &= \frac{2}{\sigma^2}\Re\left( \frac{\partial \bm{f}(\bm{\chi},\bm{\beta})^{\mathsf H}}{\partial \bm{\beta}} \frac{\partial \bm{f}(\bm{\chi},\bm{\beta})}{\partial \bm{\beta}} \right) \notag\\
		&= \frac{2}{\sigma^2}\Re\left( (\bm{b}(\tilde{\bm{r}},\bm{\chi}) [1,j])^{\mathsf H} \bm{b}(\tilde{\bm{r}},\bm{\chi}) [1,j] \right) \notag\\
		&= \frac{2}{\sigma^2}\Re\left( [1,j]^{\mathsf H} {\rm{Tr}}\left(\bm{\alpha}\bm{\alpha}^{\mathsf T} \bm{S}^{\rm{s}} (\bm{S}^{\rm{s}})^{\mathsf{H}}\bm{\alpha}^{\mathsf *}\bm{\alpha}^{\mathsf H} \right) [1,j] \right) \notag\\
		&= \frac{2P^{\rm{s}}TN}{\sigma^2}\Re\left( \begin{bmatrix}
			1  & j \\
			-j  & 1
		\end{bmatrix} \right) = \frac{2P^{\rm{s}}TN}{\sigma^2}\bm{I}_2.
	\end{align}
	Then, $\bm{\Lambda}$ in \eqref{Lambda} can be further expressed as
	\begin{align}
		\bm{\Lambda} &= \left[\bm{J}_{\bm{\chi},\bm{\chi}} - \bm{J}_{\bm{\chi},\bm{\beta}} \bm{J}_{\bm{\beta},\bm{\beta}}^{-1} \bm{J}_{\bm{\chi},\bm{\beta}}^{\mathsf T}\right]^{-1} \\
		&= \Bigg[\frac{16\pi^2P^{\rm{s}}T|\beta|^2}{\sigma^2\lambda^2N}\begin{bmatrix}
			N\bar{\gamma}_{y,y} + \gamma_{y}\gamma_{y} & N\bar{\gamma}_{y,z} + \gamma_{y}\gamma_{z} \\
			N\bar{\gamma}_{y,z} + \gamma_{y}\gamma_{z} & N\bar{\gamma}_{z,z} + \gamma_{z}\gamma_{z}
		\end{bmatrix} \notag\\
		&~~~~~~ - \frac{\sigma^2}{2P^{\rm{s}}TN}\left(\frac{8\pi P^{\rm{s}}T}{\sigma^2\lambda}\right)^2\bm{\gamma}\bar{\bm{\beta}}^{\mathsf T} \bar{\bm{\beta}}\bm{\gamma}^{\mathsf T}\Bigg]^{-1} \notag\\
		&= \frac{\sigma^2\lambda^2N}{16\pi^2P^{\rm{s}}T|\beta|^2} \Bigg[\begin{bmatrix}
			N\bar{\gamma}_{y,y} + \gamma_{y}\gamma_{y} & N\bar{\gamma}_{y,z} + \gamma_{y}\gamma_{z} \\
			N\bar{\gamma}_{y,z} + \gamma_{y}\gamma_{z} & N\bar{\gamma}_{z,z} + \gamma_{z}\gamma_{z}
		\end{bmatrix} \notag\\
		&~~~~~~ - 2\begin{bmatrix}
			\gamma_{y}\gamma_{y} & \gamma_{y}\gamma_{z} \\
			\gamma_{y}\gamma_{z} & \gamma_{z}\gamma_{z}
		\end{bmatrix}\Bigg]^{-1} \notag\\
		&= \frac{\sigma^2\lambda^2N}{16\pi^2P^{\rm{s}}T|\beta|^2}\begin{bmatrix}
			N\bar{\gamma}_{y,y} - \gamma_{y}\gamma_{y} & N\bar{\gamma}_{y,z} - \gamma_{y}\gamma_{z} \\
			N\bar{\gamma}_{y,z} - \gamma_{y}\gamma_{z} & N\bar{\gamma}_{z,z} - \gamma_{z}\gamma_{z}
		\end{bmatrix} ^{-1} \notag\\
		&\overset{(e_1)}= \frac{\sigma^2\lambda^2N}{16\pi^2P^{\rm{s}}T|\beta|^2}\begin{bmatrix}
			N^2{\rm{var}}(\bm{y}) & N^2{\rm{cov}}(\bm{y},\bm{z}) \\
			N^2{\rm{cov}}(\bm{y},\bm{z}) & N^2{\rm{var}}(\bm{z})
		\end{bmatrix} ^{-1} \notag\\
		&\overset{(e_2)}= \frac{\sigma^2\lambda^2}{16\pi^2P^{\rm{s}}TN|\beta|^2}\frac{1}{{\rm{var}}(\bm{y}){\rm{var}}(\bm{z})-{\rm{cov}}(\bm{y},\bm{z})^2} \notag\\
		&~~~~~~\begin{bmatrix}
			{\rm{var}}(\bm{z}) & -{\rm{cov}}(\bm{y},\bm{z}) \\
			-{\rm{cov}}(\bm{y},\bm{z}) & {\rm{var}}(\bm{y})
		\end{bmatrix},\notag
	\end{align}
	where the equality $(e_1)$ holds due to the definition of the variance function ${\rm{var}}(\bm{y})$ and covariance function ${\rm{cov}}(\bm{y},\bm{z})$ in \eqref{CRBr}, the equality $(e_2)$ holds since $\begin{bmatrix}
		a & b \\
		c & d
	\end{bmatrix}^{-1}=\frac{1}{ad-bc}\begin{bmatrix}
		d & -b \\
		-c & a
	\end{bmatrix}$. Finally, the CRB expression of $\bm{\chi}$ is given by
	\begin{align}\label{CRB_app}
		{\rm{CRB}}_u(\tilde{\bm{r}})=\bm{\Lambda}[1,1] = \frac{\sigma^2\lambda^2}{16\pi^2P^{\rm{s}}TN|\beta|^2}\frac{1}{{\rm{var}}(\bm{y})-\frac{{\rm{cov}}(\bm{y},\bm{z})^2}{{\rm{var}}(\bm{z})}},\notag\\
		{\rm{CRB}}_v(\tilde{\bm{r}})=\bm{\Lambda}[2,2] = \frac{\sigma^2\lambda^2}{16\pi^2P^{\rm{s}}TN|\beta|^2}\frac{1}{{\rm{var}}(\bm{z})-\frac{{\rm{cov}}(\bm{y},\bm{z})^2}{{\rm{var}}(\bm{y})}}.
	\end{align}
	This thus completes the proof of Theorem 2.

	\bibliographystyle{IEEEtran}
	\bibliography{IEEEabrv,IEEEexample}

\begin{thebibliography}{10}
\providecommand{\url}[1]{#1}
\csname url@samestyle\endcsname
\providecommand{\newblock}{\relax}
\providecommand{\bibinfo}[2]{#2}
\providecommand{\BIBentrySTDinterwordspacing}{\spaceskip=0pt\relax}
\providecommand{\BIBentryALTinterwordstretchfactor}{4}
\providecommand{\BIBentryALTinterwordspacing}{\spaceskip=\fontdimen2\font plus
\BIBentryALTinterwordstretchfactor\fontdimen3\font minus
  \fontdimen4\font\relax}
\providecommand{\BIBforeignlanguage}[2]{{%
\expandafter\ifx\csname l@#1\endcsname\relax
\typeout{** WARNING: IEEEtran.bst: No hyphenation pattern has been}%
\typeout{** loaded for the language `#1'. Using the pattern for}%
\typeout{** the default language instead.}%
\else
\language=\csname l@#1\endcsname
\fi
#2}}
\providecommand{\BIBdecl}{\relax}
\BIBdecl

\bibitem{jiang2021the}
W.~Jiang, B.~Han, M.~A. Habibi, and H.~D. Schotten, ``{The road towards 6G: A
  comprehensive survey},'' \emph{{IEEE} Open J. Commun. Soc.}, vol.~2, pp.
  334--366, Feb. 2021.

\bibitem{chowdhury20206g}
M.~Z. Chowdhury, M.~Shahjalal, S.~Ahmed, and Y.~M. Jang, ``{6G wireless
  communication systems: Applications, requirements, technologies, challenges,
  and research directions},'' \emph{{IEEE} Open J. Commun. Soc.}, vol.~1, pp.
  957--975, 2020.

\bibitem{liu2022survey}
A.~Liu, Z.~Huang, M.~Li, Y.~Wan, W.~Li, T.~X. Han, C.~Liu, R.~Du, D.~K.~P. Tan,
  J.~Lu, Y.~Shen, F.~Colone, and K.~Chetty, ``{A survey on fundamental limits
  of integrated sensing and communication},'' \emph{{IEEE} Commun. Surveys
  Tuts.}, vol.~24, no.~2, pp. 994--1034, 2nd Quart., 2022.

\bibitem{shao2024intelligent}
X.~Shao, C.~You, and R.~Zhang, ``{Intelligent reflecting surface aided wireless
  sensing: Applications and design issues},'' \emph{{IEEE} Wireless Commun.},
  early access, 2024, doi: 10.1109/MWC.004.2300058.

\bibitem{mailloux2005phased}
R.~J. Mailloux, \emph{{Phased Array Antenna Handbook}}.\hskip 1em plus 0.5em
  minus 0.4em\relax 2nd ed. Norwood, MA, USA: Artech House, 2005.

\bibitem{wirth2005radar}
W.-D. Wirth, \emph{{Radar Techniques Using Array Antennas}}.\hskip 1em plus
  0.5em minus 0.4em\relax 2nd ed. Edison, NJ, USA: IET, 2005.

\bibitem{greene1978sparse}
C.~R. Greene and R.~C. Wood, ``{Sparse array performance},'' \emph{J. Acoust.
  Soc. Am.}, vol.~63, no.~6, pp. 1866--1872, Feb. 1978.

\bibitem{roberts2011sparse}
W.~Roberts, L.~Xu, J.~Li, and P.~Stoica, ``{Sparse antenna array design for
  MIMO active sensing applications},'' \emph{{IEEE} Trans. Antennas Propagat.},
  vol.~59, no.~3, pp. 846--858, Mar. 2011.

\bibitem{wang2023can}
H.~Wang, C.~Feng, Y.~Zeng, S.~Jin, C.~Yuen, B.~Clerckx, and R.~Zhang,
  ``Enhancing spatial multiplexing and interference suppression for near-and
  far-field communications with sparse $\text{MIMO}$,'' \emph{arXiv preprint
  arXiv:2408.01956}, 2024.

\bibitem{gazzah2009optimum}
H.~Gazzah and K.~Abed-Meraim, ``{Optimum ambiguity-free directional and
  omnidirectional planar antenna arrays for DOA estimation},'' \emph{{IEEE}
  Trans. Signal Processing}, vol.~57, no.~10, pp. 3942--3953, Oct. 2009.

\bibitem{zhu2023MAMag}
L.~Zhu, W.~Ma, and R.~Zhang, ``Movable antennas for wireless communication:
  Opportunities and challenges,'' \emph{IEEE Commun. Mag.}, vol.~62, no.~6, pp.
  114--120, June 2024.

\bibitem{zhu2024historical}
L.~Zhu and K.-K. Wong, ``Historical review of fluid antenna and movable
  antenna,'' \emph{arXiv preprint arXiv:2401.02362}, 2024.

\bibitem{new2024tutorial}
W.~K. New, K.-K. Wong, H.~Xu, C.~Wang, F.~R. Ghadi, J.~Zhang, J.~Rao, R.~Murch,
  P.~Ram{\'\i}rez-Espinosa, D.~Morales-Jimenez, C.-B. Chae, and K.-F. Tong, ``A
  tutorial on fluid antenna system for {6G} networks: Encompassing
  communication theory, optimization methods and hardware designs,'' \emph{IEEE
  Commun. Surveys Tuts.}, 2024, early access, DOI: 10.1109/COMST.2024.3498855.

\bibitem{wutuo2024fluid}
T.~Wu, K.~Zhi, J.~Yao, X.~Lai, J.~Zheng, H.~Niu, M.~Elkashlan, K.-K. Wong,
  C.-B. Chae, Z.~Ding, G.~K. Karagiannidis, M.~Debbah, and C.~Yuen, ``Fluid
  antenna systems enabling {6G}: Principles, applications, and research
  directions,'' \emph{arXiv preprint arXiv:2412.03839}, 2024.

\bibitem{yeyuqi2023fluid}
Y.~Ye, L.~You, J.~Wang, H.~Xu, K.-K. Wong, and X.~Gao, ``{Fluid
  antenna-assisted MIMO transmission exploiting statistical CSI},'' \emph{IEEE
  Commun. Lett.}, vol.~28, no.~1, pp. 223--227, Jan. 2024.

\bibitem{ma2024MAsensing}
W.~Ma, L.~Zhu, and R.~Zhang, ``Movable antenna enhanced wireless sensing via
  antenna position optimization,'' \emph{IEEE Trans. Wireless Commun.},
  vol.~23, no.~11, pp. 16\,575--16\,589, Nov. 2024.

\bibitem{zhu2023MAmultiuser}
L.~Zhu, W.~Ma, B.~Ning, and R.~Zhang, ``Movable-antenna enhanced multiuser
  communication via antenna position optimization,'' \emph{IEEE Trans. Wireless
  Commun.}, vol.~23, no.~7, pp. 7214--7229, July 2024.

\bibitem{zhu2022MAmodel}
{L. Zhu, W. Ma, and R. Zhang}, ``Modeling and performance analysis for movable
  antenna enabled wireless communications,'' \emph{IEEE Trans. Wireless
  Commun.}, vol.~23, no.~6, pp. 6234--6250, June 2024, arXiv accessed on 11
  Oct. 2022.

\bibitem{mei2024movable}
W.~Mei, X.~Wei, B.~Ning, Z.~Chen, and R.~Zhang, ``Movable-antenna position
  optimization: A graph-based approach,'' \emph{IEEE Wireless Commun. Lett.},
  vol.~13, no.~7, pp. 1853--1857, July 2024.

\bibitem{zhu2024wideband}
L.~Zhu, W.~Ma, Z.~Xiao, and R.~Zhang, ``Performance analysis and optimization
  for movable antenna aided wideband communications,'' \emph{IEEE Trans.
  Wireless Commun.}, 2024, early access, DOI: 10.1109/TWC.2024.3471698.

\bibitem{xiao2023multiuser}
Z.~Xiao, X.~Pi, L.~Zhu, X.-G. Xia, and R.~Zhang, ``Multiuser communications
  with movable-antenna base station: Joint antenna positioning, receive
  combining, and power control,'' \emph{arXiv preprint arXiv:2308.09512}, 2023.

\bibitem{wu2023movable}
Y.~Wu, D.~Xu, D.~W.~K. Ng, W.~Gerstacker, and R.~Schober, ``Movable
  antenna-enhanced multiuser communication: Optimal discrete antenna
  positioning and beamforming,'' in \emph{Proc. IEEE Global Commun. Conf.
  (Globecom)}, Kuala Lumpur, Malaysia, Dec. 2023, pp. 7508--7513.

\bibitem{qin2024antenna}
H.~Qin, W.~Chen, Z.~Li, Q.~Wu, N.~Cheng, and F.~Chen, ``Antenna positioning and
  beamforming design for fluid antenna-assisted multi-user downlink
  communications,'' \emph{IEEE Wireless Commun. Lett.}, 2024, early access,
  DOI: 10.1109/LWC.2024.3360117.

\bibitem{cheng2023sum}
Z.~Cheng, N.~Li, J.~Zhu, X.~She, C.~Ouyang, and P.~Chen, ``Sum-rate
  maximization for movable antenna enabled multiuser communications,''
  \emph{arXiv preprint arXiv:2309.11135}, 2023.

\bibitem{yang2024flexible}
S.~Yang, J.~An, Y.~Xiu, W.~Lyu, B.~Ning, Z.~Zhang, M.~Debbah, and C.~Yuen,
  ``Flexible antenna arrays for wireless communications: Modeling and
  performance evaluation,'' \emph{arXiv preprint arXiv:2407.04944}, 2024.

\bibitem{ma2022MAmimo}
W.~Ma, L.~Zhu, and R.~Zhang, ``{MIMO} capacity characterization for movable
  antenna systems,'' \emph{IEEE Trans. Wireless Commun.}, vol.~23, no.~4, pp.
  3392--3407, Apr. 2024, arXiv accessed on 11 Oct. 2022.

\bibitem{chen2023joint}
X.~Chen, B.~Feng, Y.~Wu, D.~W.~K. Ng, and R.~Schober, ``Joint beamforming and
  antenna movement design for moveable antenna systems based on statistical
  {CSI},'' in \emph{Proc. IEEE Global Commun. Conf. (Globecom)}, Kuala Lumpur,
  Malaysia, Dec. 2023, pp. 4387--4392.

\bibitem{ma2023MAestimation}
W.~Ma, L.~Zhu, and R.~Zhang, ``Compressed sensing based channel estimation for
  movable antenna communications,'' \emph{IEEE Commun. Lett.}, vol.~27, no.~10,
  pp. 2747--2751, Oct. 2023.

\bibitem{xiao2023channel}
Z.~Xiao, S.~Cao, L.~Zhu, Y.~Liu, B.~Ning, X.-G. Xia, and R.~Zhang, ``Channel
  estimation for movable antenna communication systems: A framework based on
  compressed sensing,'' \emph{IEEE Trans. Wireless Commun.}, 2024, early
  access, DOI: 10.1109/TWC.2024.3385110.

\bibitem{zhu2023MAarray}
L.~Zhu, W.~Ma, and R.~Zhang, ``Movable-antenna array enhanced beamforming:
  Achieving full array gain with null steering,'' \emph{IEEE Commun. Lett.},
  vol.~27, no.~12, pp. 3340--3344, Dec. 2023.

\bibitem{ma2024multi}
W.~Ma, L.~Zhu, and R.~Zhang, ``Multi-beam forming with movable-antenna array,''
  \emph{IEEE Commun. Lett.}, vol.~28, no.~3, pp. 697--701, Mar. 2024.

\bibitem{ZhuLP_satellite_MA}
L.~Zhu, X.~Pi, W.~Ma, Z.~Xiao, and R.~Zhang, ``Dynamic beam coverage for
  satellite communications aided by movable-antenna array,'' \emph{arXiv
  preprint arXiv:2404.15643}, 2024.

\bibitem{hu2024secure}
G.~Hu, Q.~Wu, K.~Xu, J.~Si, and N.~Al-Dhahir, ``Secure wireless communication
  via movable-antenna array,'' \emph{IEEE Signal Process. Lett.}, vol.~31, pp.
  516--520, Jan. 2024.

\bibitem{tang2024secure}
J.~Tang, C.~Pan, Y.~Zhang, H.~Ren, and K.~Wang, ``Secure {MIMO} communication
  relying on movable antennas,'' \emph{arXiv preprint arXiv:2403.04269}, 2024.

\bibitem{shao20246DMA}
X.~Shao, Q.~Jiang, and R.~Zhang, ``{6D movable antenna based on user
  distribution: Modeling and optimization},'' \emph{IEEE Trans. Wireless
  Commun.}, 2024, early access.

\bibitem{shao2024discrete}
X.~Shao, R.~Zhang, Q.~Jiang, and R.~Schober, ``{6D movable antenna enhanced
  wireless network via discrete position and rotation optimization},''
  \emph{IEEE J. Sel. Areas Commun.}, 2024, early access.

\bibitem{shao2024Mag6DMA}
X.~Shao and R.~Zhang, ``{6DMA} enhanced wireless network with flexible antenna
  position and rotation: Opportunities and challenges,'' \emph{IEEE Commun.
  Mag.}, 2024, early access.

\bibitem{shao2024exploiting}
X.~Shao, R.~Zhang, and R.~Schober, ``Exploiting six-dimensional movable antenna
  for wireless sensing,'' \emph{IEEE Wireless Commun. Lett.}, 2024, early
  access.

\bibitem{WuHS_MA_RIS_ISAC}
H.~Wu, H.~Ren, and C.~Pan, ``Movable antenna-enabled {RIS}-aided integrated
  sensing and communication,'' \emph{arXiv preprint arXiv:2407.03228}, 2024.

\bibitem{hao2024fluid}
T.~Hao, C.~Shi, Q.~Wu, B.~Xia, Y.~Guo, L.~Ding, and F.~Yang, ``{Fluid-antenna
  enhanced ISAC: Joint antenna positioning and dual-functional beamforming
  design under perfect and imperfect CSI},'' \emph{arXiv preprint
  arXiv:2407.18988}, 2024.

\bibitem{kuang2024movableISAC}
Z.~Kuang, W.~Liu, C.~Wang, Z.~Jin, J.~Ren, X.~Zhang, and Y.~Shen,
  ``Movable-antenna array empowered {ISAC} systems for low-altitude economy,''
  \emph{arXiv preprint arXiv:2406.07374}, 2024.

\bibitem{zhang2024efficient}
Q.~Zhang, M.~Shao, T.~Zhang, G.~Chen, and J.~Liu, ``{An efficient algorithm for
  sum-rate maximization in fluid antenna-assisted ISAC system},'' \emph{arXiv
  preprint arXiv:2405.06516}, 2024.

\bibitem{mayaodong2024movable}
Y.~Ma, K.~Liu, Y.~Liu, L.~Zhu, and Z.~Xiao, ``Movable-antenna aided secure
  transmission for {RIS}-{ISAC} systems,'' \emph{arXiv preprint
  arXiv:2410.03426}, 2024.

\bibitem{khalili2024advanced}
A.~Khalili and R.~Schober, ``{Advanced ISAC design: Movable antennas and
  accounting for dynamic RCS},'' \emph{arXiv preprint arXiv:2407.20930}, 2024.

\bibitem{lyu2024flexibleISAC}
W.~Lyu, S.~Yang, Y.~Xiu, Z.~Zhang, C.~Assi, and C.~Yuen, ``Flexible beamforming
  for movable antenna-enabled integrated sensing and communication,''
  \emph{arXiv preprint arXiv:2405.10507}, 2024.

\bibitem{peng2024jointISAC}
S.~Peng, C.~Zhang, Y.~Xu, X.~Ou, and D.~He, ``Joint antenna position and
  beamforming optimization with self-interference mitigation in {MA-ISAC}
  system,'' \emph{arXiv preprint arXiv:2408.00413}, 2024.

\bibitem{wang2024multiuser}
C.~Wang, G.~Li, H.~Zhang, K.-K. Wong, Z.~Li, D.~W.~K. Ng, and C.-B. Chae,
  ``Fluid antenna system liberating multiuser {MIMO} for {ISAC} via deep
  reinforcement learning,'' \emph{IEEE Trans. Wireless Commun.}, 2024, early
  access, DOI: 10.1109/TWC.2024.3376800.

\bibitem{xiu2024movable}
Y.~Xiu, S.~Yang, W.~Lyu, P.~L. Yeoh, Y.~Li, and Y.~Ai, ``{Movable antenna
  enabled ISAC beamforming design for low-altitude airborne vehicles},''
  \emph{arXiv preprint arXiv:2409.15923}, 2024.

\bibitem{qin2024cramer}
H.~Qin, W.~Chen, Q.~Wu, Z.~Zhang, Z.~Li, and N.~Cheng, ``{Cramer-Rao bound
  minimization for movable antenna-assisted multiuser integrated sensing and
  communications},'' \emph{arXiv preprint arXiv:2407.11875}, 2024.

\bibitem{lirenwang2024irs}
R.~Li, X.~Shao, S.~Sun, M.~Tao, and R.~Zhang, ``{IRS} aided millimeter-wave
  sensing and communication: Beam scanning, beam splitting, and performance
  analysis,'' \emph{arXiv preprint arXiv:2401.15344}, 2024.

\bibitem{kay1993fundamentals}
S.~M. Kay, \emph{{Fundamentals of Statistical Signal Processing: Estimation
  Theory}}.\hskip 1em plus 0.5em minus 0.4em\relax Prentice-Hall, Inc., 1993.

\bibitem{shao2022target}
X.~Shao, C.~You, W.~Ma, X.~Chen, and R.~Zhang, ``Target sensing with
  intelligent reflecting surface: Architecture and performance,'' \emph{IEEE J.
  Select. Areas Commun.}, vol.~40, no.~7, pp. 2070--2084, Jul. 2022.

\bibitem{deits2015computing}
R.~Deits and R.~Tedrake, ``{Computing large convex regions of obstacle-free
  space through semidefinite programming},'' in \emph{Algorithmic Foundations
  of Robotics XI}, vol. 107.\hskip 1em plus 0.5em minus 0.4em\relax Springer,
  Jan. 2015, pp. 109--124.

\bibitem{clarkson2010coresets}
K.~L. Clarkson, ``{Coresets, sparse greedy approximation, and the Frank-Wolfe
  algorithm},'' \emph{{ACM} Trans. Algorithms}, vol.~6, no.~4, pp. 1--30, Sep.
  2010.

\end{thebibliography}
\end{document}